\documentclass[11pt, DIV=8]{scrartcl}

\usepackage{amsmath, amsthm} % amssymb not needed with newtxmath
\usepackage{newtxtext,newtxmath}  % times font
\usepackage[scale=0.925]{sourcecodepro}  % fixed space font for code
\usepackage[final]{microtype}  % kerning, etc.
\usepackage{mleftright} % reduces extreme sizes of \left and \right
    \mleftright
\usepackage[
  paper=a4paper,
  left=1.5in,
  right=1.5in,
  top=1in,
  bottom=1.25in
]{geometry}

\usepackage{mathtools}

\usepackage[dvipsnames]{xcolor}
\usepackage{hyperref}
\hypersetup{
  breaklinks=true,
  colorlinks=true,
  allcolors=MidnightBlue
}

\usepackage{natbib}
\usepackage[nameinlink]{cleveref}

\usepackage{graphicx}
\usepackage{tikz}
    \usetikzlibrary{decorations.pathreplacing, positioning, cd}
\usepackage[font=footnotesize,labelfont=footnotesize]{caption}

\usepackage{algorithm, algpseudocode}

\mathtoolsset{showmanualtags}

\newcommand{\fapply}[2]{#1 \left( #2 \right)}

\DeclarePairedDelimiter{\norm}{\lVert}{\rVert}
\DeclarePairedDelimiterXPP\normal[1]{\mathcal{N}}{(}{)}{}{#1}
\newcommand{\kldivhat}[1]{\fapply{\widehat{D}_{\mathrm{KL}}}{#1}}
\newcommand{\fdivhatOp}{\hat{D}_{\mathrm{F}}}
\newcommand{\fdivhat}[1]{\fapply{\widehat{D}_{\mathrm{F}}}{#1}}
\newcommand{\absdet}[1]{\left| \det #1 \right|}
\newcommand{\drw}[2]{#1^{(#2)}}
\newcommand{\bigo}[1]{\fapply{\mathcal{O}}{#1}}
\newcommand{\ddx}{\frac{\partial}{\partial \x}}
\newcommand{\ddy}{\frac{\partial}{\partial \y}}
\newcommand{\Id}{\mathrm{I}}
\newcommand{\vecspan}[1]{\fapply{\operatorname{span}}{#1}}
\newcommand{\R}{\mathbb{R}}
\DeclarePairedDelimiterX\inner[2]{\langle}{\rangle}{#1, #2}
\newcommand{\dd}{\mathop{}\!\mathrm{d}}
\newcommand{\ones}{\mathrm{e}}
\newcommand{\vol}{\mathrm{vol}}

\newcommand{\diverge}[4]{#1_{#2}\left[#3 \, |\!| \, #4\right]}
\newcommand{\divergeD}[3]{\diverge{D}{#1}{#2}{#3}}
\newcommand{\kldiv}[2]{\divergeD{\textrm{KL}}{#1}{#2}}
\newcommand{\fdiv}[2]{\divergeD{\textrm{F}}{#1}{#2}}

\DeclareMathOperator{\covOp}{cov}
\newcommand{\cov}[1]{\fapply{\covOp}{#1}}

\DeclareMathOperator{\trOp}{tr}
\newcommand{\tr}[1]{\fapply{\trOp}{#1}}

\DeclareMathOperator{\diagOp}{diag}
\newcommand{\diag}[1]{\fapply{\diagOp}{#1}}

\DeclareMathOperator{\diagMatrixOp}{diag}
\newcommand{\diagMatrix}[1]{\fapply{\diagMatrixOp}{#1}}

\newcommand{\invdiag}[1]{\fapply{\diagMatrixOp^{-1}}{#1}}

\DeclareMathOperator{\varOp}{var}
\newcommand{\var}[1]{\fapply{\varOp}{#1}}

\DeclareMathOperator*{\argmin}{arg\,min}

\newcommand{\spdmean}[2]{#1 \,\#\, #2} % {\spdm(#1, #2)}

\newcommand{\innerproduct}[3]{\langle #2, #3 \rangle_{#1}}

\newcommand{\A}{\mathcal{A}}
\newcommand{\T}{\mathcal{T}}
\newcommand{\x}{x}
\newcommand{\y}{y}

\newcommand{\proofpart}[1]{\medskip\noindent\textbf{#1.}\quad}

\providecommand{\keywords}[1]
{\small	\textbf{\textit{Keywords}:} #1}

\theoremstyle{plain}
\newtheorem{definition}{Definition}[section]
\newtheorem{theorem}[definition]{Theorem}
\newtheorem{corollary}[definition]{Corollary}
\newtheorem{lemma}[definition]{Lemma}

\title{Preconditioning Hamiltonian Monte Carlo by minimizing Fisher Divergence}

\author{Adrian Seyboldt\thanks{PyMC Labs, \href{mailto:adrian.seyboldt@gmail.com}{\texttt{adrian.seyboldt@gmail.com}}}
\and
Eliot L. Carlson\thanks{Center for Computational Mathematics, Flatiron Institute, 
\href{mailto:ecarlson@flatironinstitute.org}{\texttt{ecarlson@flatironinstitute.org}}}
\and
Bob Carpenter\thanks{Center for Computational Mathematics, Flatiron Institute, \href{mailto:bcarpenter@flatironinstitute.org}{\texttt{bcarpenter@flatironinstitute.org}}}
}

\date{\today}

\begin{document}
\maketitle

\begin{abstract}
\noindent\small
Although Hamiltonian Monte Carlo (HMC) scales as $\bigo{d^{1/4}}$ in dimension, there is a large constant factor determined by the curvature of the target density. This constant factor can be reduced in most cases through preconditioning, the state of the art for which uses diagonal or dense penalized maximum likelihood estimation of (co)variance based on a sample of warmup draws. These estimates converge slowly in the diagonal case and scale poorly when expanded to the dense case. We propose a more effective estimator based on minimizing the sample Fisher divergence from a linearly transformed density to a standard normal distribution. We present this estimator in three forms, (a)~diagonal, (b)~dense, and (c)~low-rank plus diagonal.  Using a collection of 114 models from \texttt{posteriordb}, we demonstrate that the diagonal minimizer of Fisher divergence outperforms the industry-standard variance-based diagonal estimators used by Stan and PyMC by a median factor of 1.3.  The low-rank plus diagonal minimizer of the Fisher divergence outperforms Stan and PyMC's diagonal estimators by a median factor of 4.
\end{abstract}
\normalsize

\noindent

\section{Introduction}

Hamiltonian Monte Carlo (HMC) is a Markov Chain Monte Carlo (MCMC) method that is widely used in Bayesian inference for sampling from differentiable high-dimensional posterior distributions \citep{betancourt_conceptual_2018, neal2011mcmc}. By simulating Hamiltonian trajectories in a phase space consisting of a position $\x \in \mathbb{R}^d$ augmented with a momentum $\rho \in \mathbb{R}^d$, HMC can explore high-dimensional parameter spaces more efficiently both in theory and in practice than gradient-free MCMC techniques such as random walk Metropolis and Gibbs sampling.  HMC's scalability makes it popular in probabilistic programming libraries based on automatic differentiation, like PyMC \citep{salvatier2016probabilistic}, Stan \citep{carpenter_stan_2017}, NumPyro \citep{phan2019composable} and Turing.jl \citep{ge2018turing}.  However, HMC struggles with the stiff posteriors induced by multiscale distributions, especially when the scaling is induced by correlation rather than through independent parameters and when the eigenvectors varies by position.

\subsection{Preconditioning Hamiltonian Monte Carlo}

A common approach to mitigating HMC's sensitivity to target geometry is to \emph{precondition} the original density using a parameter transform---this is what all the probabilistic programming languages do, for example.  After sampling, the preconditioned sample must be inverse transformed to recover a sample of the parameters of interest.

The most common approach to preconditioning is through a positive-definite \emph{mass matrix} $M$, whereby in the leapfrog algorithm, the auxiliary momentum variables $\rho$ are resampled from $\normal{0, M}$ instead of the default isotropic $\normal{0, \Id}$, and the momentum is transformed by $M^{-1}$ before the position updates to match \citep{neal2011mcmc}.  Following \citet{betancourt_conceptual_2018}, the HMC samplers of all of the probabilistic programming languages in wide use precondition with a mass matrix that targets the inverse of a regularized sample covariance estimate from warmup draws. The mass matrix is taken to be diagonal by default to avoid the high computational cost of using a dense mass matrix (cubic time to invert and factor for momentum resampling, and quadratic per leapfrog step in the Hamiltonian simulations). The diagonal approach aims to rescale the dimensions such that each has unit variance. 

This is often inadequate for highly correlated target geometries or those with varying curvature. In these cases, practitioners must manually reparameterize models to achieve efficient sampling. Even where possible, this process is time-consuming and requires expertise. In most cases, the optimal reparameterization depends on the data, meaning the same model requires different parameterizations when fitted to new data.  For example, depending on how informative the prior and data are, hierarchical models for varying effects, of which there are typically several in applied models, might require centered or non-centered parameterizations \citep{papaspiliopoulos2007general,betancourt2015hamiltonian}.

\subsection{Fisher HMC}

\emph{Fisher HMC}, introduced here, adapts a transformation of the target distribution with the objective of minimizing the Fisher divergence from the transformed distribution to a standard normal.  In this paper, we focus on affine transformations, which can be handled via a mass matrix; see \Cref{appendix:transformed-hmc} for a more general perspective.  As well as using warmup draws as in standard adaptation, Fisher HMC uses the scores of the warmup draws (i.e., the derivatives of the target log density evaluated at the warmup draws).  The scores require no additional computation because they are required for the Hamiltonian simulation step of HMC. We include specific algorithms for three classes of positive-definite mass matrices: diagonal, dense, and low rank plus diagonal.

\subsection{Summary of contributions}

Our primary contributions are as follows.
\begin{enumerate}
    \item We introduce a novel mass-matrix adaptation objective---minimizing
     the Fisher divergence from the preconditioned target density to a standard normal.%
     \footnote{ \Cref{appendix:transformed-hmc} provides the geometric motivation for Fisher divergence.}
    \item We provide an analysis of the resulting condition numbers for multivariate normal targets.
    \item We introduce an efficient scheme for estimating mass matrices using warmup draws and their scores and apply it to diagonal, dense, and low-rank plus diagonal mass matrices.
    \item We experimentally demonstrate that \texttt{nutpie}'s Fisher HMC outperforms the No-U-Turn Sampler's covariance-based mass matrix estimates on a wide range of model classes.
\end{enumerate}

\section{Theoretical motivation for Fisher HMC}

HMC is a gradient-based method, meaning that the algorithm computes the partial derivatives of the log posterior density, or \emph{scores}.  Given a target density $p(\x)$, its score function is defined by 
$$
S(\x) = \ddx \log p(\x).
$$
While these scores contain useful information about the target density, the computationally tractable preconditioning methods in wide use \citep{carpenter_stan_2017, bales2019} ignore them.  We discuss why and how this information can be valuable for constructing a preconditioner for HMC sampling.

\subsection{An example with a normal target density}
\label{sec: motivation_normal}

To illustrate how valuable scores are, consider a univariate normal target density $p(\x) = \normal{\x \mid \mu, \sigma^2}$ for $\x, \mu \in \mathbb{R}$ and $\sigma \in (0, \infty)$.  The log density for a univariate normal is
$$
\log p(x) \propto -\frac{1}{2} \left( \frac{x - \mu}{\sigma} \right)^2,
$$
and the score function is
$$
S(\x) = -\frac{x - \mu}{\sigma^2}.
$$
Suppose we have a pair of distinct samples, $\drw{\x}{1},\drw{\x}{2} \sim \normal{\mu, \sigma^2}$ along with their corresponding scores $\drw{\alpha}{1} = S(\drw{x}{1})$ and $\drw{\alpha}{2} = S(\drw{\x}{2})$.  A little algebra shows that two draws are sufficient to identify the true parameters as
\begin{equation}\label{eq:diag-solution}
\mu = \overline{\x} + \sigma^2 \, \overline{\alpha}
\qquad
\sigma^2 = \sqrt{\var{x} / \var{\alpha}}.
\end{equation}
where $\overline{\x}, \var{\x}$ and $\overline{\alpha}, \var{\alpha}$ are the means and variances of the vectors $x = [\drw{x}{1} \ \drw{x}{2}]$ and $\alpha = [\drw{\alpha}{1} \ \drw{\alpha}{2}]$, respectively.%
\footnote{The answer is the same with variance taken to be the maximum likelihood estimate or unbiased estimate.}

Estimators based solely on samples from a distribution are always limited by the Cramér–Rao bound.  In contrast, estimators that incorporate sample score information can bypass this bound and converge at a faster rate (e.g., the two steps shown above for the univariate normal).  Later, we will see that for general multivariate distributions, an exact covariance estimate requires a number of draws equal to one plus the dimension of the distribution.

\subsection{Fisher divergence}
\label{sec: fisher_divergence}

Let $p_\T(\x)$ be a density on the \emph{target space} $\T$ from which we want to sample. We apply a change of variables via the diffeomorphism \[F\colon \A \to \T\] that transforms samples from an auxiliary \emph{adapted space} $\A$ onto $\T$\footnote{Since $F$ is a diffeomorphism, the direction of the map does not change any results. We define it from the adapted space to the target space, in line with the literature on normalizing flows.}.
Here $\T$ and $\A$ denote the same ambient space $\R^d$, but we maintain distinct notation to emphasize their different roles.

By the change-of-variables formula, the induced density on $\A$ is
\[
p_\A(\y; F) = \fapply{p_\T}{F(\y)} \, \absdet{J_F(\y)},
\]
where $J_F$ is the Jacobian of $F$ and $\absdet{J_F(\y)}$ its absolute determinant. When the choice of $F$ is clear from context, we omit it from the notation and write $p_\A(\y)$ for $p_\A(\y; F)$.

Given an affine transformation $F(\y) = A\y + \mu$ with $(AA^\top)^{-1} = M$, sampling from $p_\A$ is equivalent to sampling from $p_\T$ with mass matrix $M$ \citep{neal2011mcmc}.\footnote{The choice of matrix square root here does not matter.}

We would like to choose a preconditioner $F^*$ from a rich enough family $\mathcal{F}$ of transformations that the transformed density $p_\A$ will be close a standard normal distribution.  Any measure of divergence $D$ from $p_\A$ to $\normal{0, \Id}$ can be used to define an optimal transform $F^*$ drawn from the family $\mathcal{F}$ of transforms,
\begin{equation}
    \label{eq:divergence}
    F^* = \argmin_{F\in \mathcal{F}} \ \diverge{D}{}{p_\A(\cdot\,; F)}{\normal{0, I}}.
\end{equation}
In practice, we will use a Monte Carlo estimate of the divergence based on warmup samples from the target distribution $p_\T$. This objective is reminiscent of variational inference, where one minimizes $\diverge{D}{}{\normal{0, \Id}}{p_\A}$ \citep{wainwright2008graphical}. Here, however, we take the divergence in the opposite direction, $\diverge{D}{}{p_\A}{\normal{0, \Id}}$.

The most widely applied divergence is the Kullback-Leibler (KL)
divergence.  The KL divergence from $p$ to $q$ measures the cost of using density $q$ to encode outcomes from distribution $p$:
\[
\kldiv{p}{q}
= \int \! p(\y) \, \log \left(\frac{p(\y)}{q(\y)}\right) \, \textrm{d}\y.
\]
The divergence is always non-negative, equal to zero if and only if the arguments match, but it is not symmetric.

Minimizing the KL divergence leads to the choice of inverse mass matrix equal to the covariance of the draws, $M^{-1} = \cov{x}$ or, if restricted to be diagonal, $M^{-1} = \diagMatrix{\var{x}}$; see \Cref{proof:kl-divergence-min} for a proof. \citet{tran_tuning_2024} minimized the divergence from the distribution of the scores $p_\T$ to a standard normal distribution, which results in a mass matrix estimate equal to the variance of the scores, $M = \var{\alpha}$.

Minimizing KL divergence for mass matrix estimation using either scores or draws alone is suboptimal even for Gaussian targets, because it does not penalize both large and small eigenvalues of the covariance matrix equally; see \Cref{condition_number} for more details.  Practically speaking, this is because it concentrates on minimizing divergence to either the score covariance or draw covariance, but does not use the information from both simultaneously. We propose using Fisher divergence instead of KL divergence.  Fisher divergence uses information from both the draws and scores to condition both large and small eigenvalues.  The definition of Fisher divergence is as follows; \Cref{appendix:transformed-hmc} contains a more general geometric definition.

\begin{definition}[Fisher divergence]
Let $p$ and $q$ be two differentiable probability densities supported in $\mathbb{R}^d$, and let $\norm{v}_{G^{-1}} = (v^TG^{-1} v)^{1/2}$ be a norm, where $G$ is a symmetric positive definite matrix. The Fisher divergence from $p$ to $q$ is
\[
\fdiv{p}{q}
= \int \! p(\y) \, \norm*{\ddy\log\left(\frac{q(\y)}{p(\y)}\right)}_{G^{-1}}^2 \, \textrm{d} \y.
\]
\end{definition}
\noindent
Since scores transform contravariantly, we use the inverse of $G$ in the norm.

While we don't have an inner product on our target space $\T$ yet, we do have a natural norm on the adapted space $\A = F^{-1}(\T)$, because the standard normal $q$ is defined with respect to an inner product. This makes the standard Euclidean norm on the adapted space a natural choice. Suppose our target density is $p_\T(\x)$. Our divergence transforms back to the original target space $\T$ as
\begin{align*}
    \fdiv{p_\A}{q}
    &= \int_\A p_\A(\y) \norm*{\ddy\log\left(\frac{q(\y)}{p_\A(\y)}\right)}^2 \textrm{d}\y \\
    &= \int_\T p_\T(\x) \norm*{\ddx\log\left(\frac{q(\x)}{p_\T(\x)}\right)}^2_{\left(J_F(F^{-1}(\x)) J_F(F^{-1}(\x))^\top\right)^{-1}} \textrm{d}\x \\
    &= \int_\T p_\T(\x) \norm*{\ddx\log\left(\frac{q(\x)}{p_\T(\x)}\right)}^2_{M} \textrm{d} \x,
\end{align*}
where $J_F(y)$ is the Jacobian of the transformation $F\colon \A \to \T$. 

In practice, the Fisher divergence is only estimable from a finite number of posterior draws $\drw{\x}{i}$ and their scores, $\drw{\alpha}{i} = S(\drw{\x}{i})$. With this finite sample of draws, we use the Monte Carlo estimate $\widehat{D}$ of $\fdiv{p_\A}{\normal{0, \Id}}$,
\begin{equation*}
\widehat{D} = \frac{1}{N} \sum_{i=1}^N \, \norm*{\,\drw{\beta}{i} + \drw{\y}{i}}^2,
\end{equation*}
where $\drw{\y}{i}$ and $\drw{\beta}{i}$ are the samples and scores in the adapted space $\A$,
$$
    \drw{\y}{i} = F^{-1}\big(\drw{\x}{i}\big) 
    \qquad \qquad
    \drw{\beta}{i} = J_F^\top\big(\drw{\y}{i}\big) \, \alpha^{(i)} + \ddy \log \, \absdet{J_F\big(\drw{\y}{i}\big)}.
$$

\subsection{Condition number for Gaussian target densities}
\label{condition_number}

Although most of the target distributions from which we would like to sample are not strictly Gaussian, posterior distributions in Bayesian inference are often log-concave, particularly when working with large data sets, due to the Bernstein–von Mises theorem. Most of the target densities in \texttt{posteriordb} are log concave; those that aren't have positive-definite Hessians throughout most of their posterior probability mass \citep{magnusson_posteriordb_2024}. These densities are still relatively easy to analyze, so we will focus on them in this section.

HMC's efficiency is limited by the ratio of the largest to the smallest eigenvalues of the Hessian of the negative log density of the target, i.e., the condition number.  Small scales constrain the step size to be small in order to remain stable, whereas large scales require many steps to complete an orbit through the target space. The leapfrog algorithm is stable only when the step size is smaller than $2 / \sqrt{\max(\lambda)}$, where $\lambda$ is the spectrum (i.e., the vector of eigenvalues) of the Hessian of $-\log p(\x)$ evaluated at the current position $\x$ \citep{leimkuhler2004simulating}.  The condition number of a matrix is defined to be the ratio of its largest to its smallest eigenvalue, $\kappa = \max(\lambda) / \min(\lambda)$.  The larger the condition number, the more steps are required to complete an orbit in the largest length-scale component of the target density.

\begin{figure}[t!]
    \centering
    \includegraphics[width=0.8\textwidth]{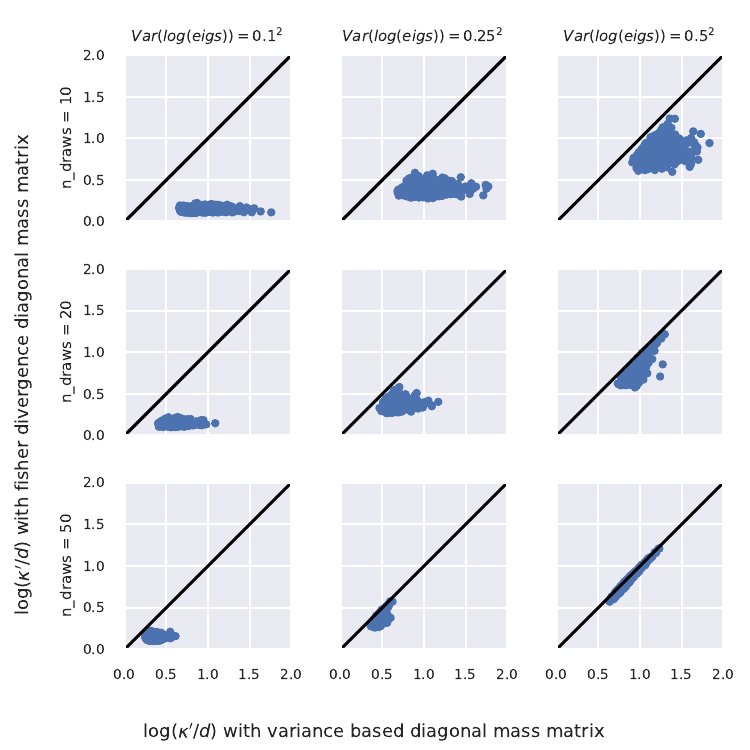}
    \caption{\textit{
            Simulated log condition numbers $\kappa$ for Fisher- and covariance-based diagonal preconditioners. For each experiment, we generate 1000 random covariance spectra ($\Sigma = D^{1/2}U\diag{\lambda^2}U^TD^{1/2}$, $U\sim \text{Uniform}((0, 200))$, $\log(\lambda_i)\sim \normal{0, \sigma^2}$, $\log(D_{i,i})\sim \normal{0, 2^2}$) and simulate sets of [10, 20, 50] i.i.d. draws from $\normal{0, \Sigma}$. We then calculate the condition number $\kappa'$ of the adapted posterior with variance and Fisher-divergence based diagonal preconditioning.
    }}
    \label{fig:condition_number}
\end{figure}

For multivariate Gaussian target densities we can be more precise about the efficiency. Suppose we have a centered multivariate Gaussian target density $\normal{\x \mid \mu, \Sigma}$ with eigenvalues $\lambda$ for the Hessian of $-\log \normal{\x \mid \mu, \Sigma}$. \citet{langmore_condition_2020} showed that the number of leapfrog steps per iteration required for efficient Hamiltonian Monte Carlo is governed by a condition-number-like quantity defined by%
\footnote{Our definition matches that of \citet{langmore_condition_2020}, despite the difference in exponents. We use an exponent of 2 because we take the eigenvalues of $\Sigma = AA^\top$, whereas they use an exponent of 4 in their definition (3.1) because they take the eigenvalues of $A$.}
\begin{equation}\label{eq:kappa}
\kappa' 
= \left(
\sum_{i=1}^d \left(\frac{\max(\lambda)}{\lambda_i}\right)^2
\right)^{1/4},
\end{equation}
where $\lambda$ is the vector of eigenvalues of the Hessian of the negative log density, $-\frac{\partial^2}{\partial\x^2} \log p(\x)$. In the following, we analyze how different divergence choices penalize the eigenspectrum of multivariate normal target distributions.

Using the KL divergence on draws only conditions large eigenvalues, as can be seen from the the KL divergence from a centered multivariate Gaussian with covariance $\Sigma$ to a standard Gaussian,
\begin{equation}
\kldiv{\normal{0, \Sigma}}{\normal{0, \textrm{I}}}
= \frac{1}{2} \sum_{i=1}^d\left(\lambda_i - \log(\lambda_i) - 1\right).
\end{equation}
If instead we focus entirely on the scores, as \citet{tran_tuning_2024} suggest, it conditions small eigenvalues,
\begin{align*}
\kldiv{S(\normal{0, \Sigma})}{\normal{0, \Id}}
&= \kldiv{\normal{0, \Sigma^{-1}}}{\normal{0, \Id}} \\
&= \frac{1}{2} \sum_{i=1}^d
  \left( \lambda_i^{-1} - \log(\lambda_i^{-1}) - 1 \right).
\end{align*}
The Fisher divergence penalizes large and small values equally, leading to more favorable values of $\kappa$ and $\kappa'$.  The Fisher divergence from a centered multivariate Gaussian with spectrum $\lambda$ to a standard Gaussian is
\begin{align*}
\label{eq:fdiv-gaussian}
    \fdiv{\normal{0, \Sigma}}{\normal{0, \textrm{I}}} 
    &= \int \x^\top(I - \Sigma^{-1})^2 \x \normal{\x \mid 0, \Sigma}
            \, \textrm{d} \x \\
    &= \textrm{tr}\left((\textrm{I} - \Sigma^{-1})^2 \, \Sigma\right) \\
    &= \sum_{i=1}^{d} \left(\lambda_i + \lambda_i^{-1} - 2\right).
\end{align*}
We compare simulated $\kappa'$ for various Gaussian targets preconditioned with \texttt{nutpie}'s diagonal estimate and with the draws-based one in \Cref{fig:condition_number}.

\subsection{Affine transformations}

This section presents results for three families of affine transformations, (1) diagonal, (2) diagonal plus low rank, and (3) dense.  This section shows how warmup draws and their scores can be used to analytically derive the mass matrix that minimizes estimated Fisher divergence to a standard normal within the chosen family.  Table~\ref{tab:bigo} shows the computational costs in terms of time and memory for the three approaches.

\begin{table}[t]
\[
\begin{array}{r|cccc}
\textit{mass matrix} & \textit{refresh} & \textit{leapfrogs} & \textit{memory} & \textit{estimation} \\
\hline
\text{dense} & \bigo{d^3} & \bigo{d^2} & \bigo{d^2} & \bigo{k\,d^2} \\
\text{low-rank + diagonal} & \bigo{r\, d} & \bigo{r \, d} & \bigo{k \, d  + k^2} & \bigo{k^3d} \\
\text{diagonal} & \bigo{d} & \bigo{d} & \bigo{d} & \bigo{d}^*\\
\end{array}
\]
\caption{\textit{Computational cost for three families of mass matrices in dimension $d$, low-rank approximation size $r \leq d$ and window size $k$.  Refreshes are required in order to refresh momentum at the start of each HMC trajectory. The estimation cost applies only periodically during warmup, when a new mass matrix estimate is needed. The estimation cost of the diagonal case is $\bigo{d}$ amortized over the updates, a single estimation would cost $\bigo{k\,d}$.}}\label{tab:bigo}
\end{table}

Each family can be characterized as a set of affine diffeomorphisms $F: \A \to \T$.  In the most general, dense case, we have affine transforms $F_{\Sigma, \mu}: \mathbb{R}^d \rightarrow \mathbb{R}^d$, defined by
$$
F_{\Sigma, \mu}(y) = A \, y + \mu,
$$
where $\mu \in \mathbb{R}^d$ and $\Sigma \in \mathbb{R}^{d \times d}$ is symmetric and positive definite with the Cholesky factorization $\Sigma = A \, A^\top$. Because HMC is translation invariant, using $F_{\Sigma, \mu}$ to precondition the density is equivalent to using the mass matrix $M = \Sigma^{-1}$ \citep{neal2011mcmc}.  Although the location (i.e., translation) vector $\mu$ does not affect the conditioning of HMC, it can be helpful for initialization if known approximately in advance.

\subsubsection{Diagonal transformations}
\label{sec: diag_mass_matrix}

For a diagonal transform, $\Sigma$ is restricted to the diagonal form $\Sigma = \diag{\sigma^2}$, where $\sigma \in (0, \infty)^d$ is a vector of strictly positive scales and $\sigma^2$ is evaluated elementwise.  This yields the transform
$$
F_{\diag{\sigma^2}, \mu}(y) = \sigma \odot \y + \mu,
$$
where $\odot$ is elementwise product. Given draws $\drw{x}{i}$ from the target density $p_\T$ and their corresponding scores $\drw{\alpha}{i}$ for $i \in \{ 1, \ldots, n \}$, the estimate of the Fisher divergence $\fdiv{p_\A}{q}$ is
\begin{equation}
\label{eq:fisher_divergence_diag}
   \widehat{D} = \frac{1}{n} \sum_{i=1}^n \norm{\sigma \odot \drw{\alpha}{i} + \sigma^{-1} \odot (\drw{\x}{i} - \mu)}^2.
\end{equation}

\begin{theorem}[Diagonal divergence minimization]
The estimated divergence is minimized when $\mu, \sigma^2 = \mu^*, {\sigma^2}^*$, where
$$
\mu^* = \overline{\x} + {\sigma^2}^* \odot \overline{\alpha}
\qquad \qquad
{\sigma^2}^* = \diag{\invdiag{\cov{x}}} \,\#\ \left( \diag{\invdiag{\cov{\alpha}}}\right)^{-1},
$$
where $\diagMatrixOp^{-1}:\mathbb{R}^{d \times d} \rightarrow \mathbb{R}^d$ extracts the diagonal of a matrix as a vector and $\#$ is the geometric mean operator in the affine-invariant Riemannian metric over symmetric, positive-definite matrices.
\end{theorem}
\begin{proof}
See Corollary~\ref{thm:diagonal-geometric-mean} for a derivation of the geometric mean operator, and \Cref{appendix:diag-proof} for a proof.
\end{proof}
In words, the estimated minimizer $\sigma^*$ is derived by taking the geometric mean of two diagonal matrices: the diagonal of the covariance matrix of draws and the diagonal of the covariance matrix of their inverse scores.  In the normal case, the estimated mean and scale recover the analytical result from the motivating example.  

The geometric mean of the diagonal matrices reduces to
$$
{\sigma^2}^* = \sqrt{\invdiag{\cov{x}} \, / \, \invdiag{\cov{\alpha}}},
$$
where the division and square root are read elementwise; see Theorem~\ref{thm:diagonal-geometric-mean}.  The transformation corresponds to the inverse mass matrix
\begin{equation}\label{eq:diag_solution}
M^{-1} = \diag{{\sigma^2}^*}.
\end{equation}

The diagonal approximation is popular due to its reduced computational cost compared to dense or even low rank plus diagonal.  It is the default approach used in \texttt{nutpie} and all of the widely-used HMC-based samplers. Both the mass-matrix solves required during adaptation and the mass-matrix times momentum multiplies required during the leapfrog steps are $\bigo{d}$ in time and memory;  see Figure~\ref{tab:bigo} for a comparison to the dense and low-rank plus diagonal families.  

For adaptation during the warmup phase, mass matrix estimation requires $\bigo{n \, d}$ operations in $d$ dimensions and $n$ warmup draws. Because the solution only depends on the variances of the sequences of samples and scores, these can be calculated online in $\bigo{d}$ memory with stable arithmetic using Welford's algorithm \citep{welford1962};  see \Cref{appendix:welford} for more details.  

\subsubsection{Dense transformations}
\label{sec: dense_case}
The affine transform $F_{\Sigma, \mu}(y) = A \, y + \mu$ for positive-definite $\Sigma$ corresponds to the mass matrix $M = (AA^\top )^{-1}.$ The estimate of the Fisher divergence $\fdiv{p_\A}{q}$ is
\begin{equation}
    \label{divergence:dense_case}
    \widehat{D}  = \frac{1}{N} \sum_{i=1}^N \norm{A^\top \drw{\alpha}{i} + A^{-1} (\drw{\x}{i} - \mu)}^2.
\end{equation}

\begin{theorem}[Dense minimization]
The estimated Fisher divergence $\widehat{D}$ is minimized over $A$ when 
\begin{equation}
\label{eq:dense_equation}
    A \, A^\top \! \cov{\alpha} \, A \, A^\top = \cov{x}.
\end{equation}
\end{theorem}
\begin{proof}
See \Cref{proof:fisher_divergence_min}.
\end{proof}
Because $\fdiv{p_\A}{\normal{0, \Id}}$ only depends on $AA^\top$ and $\mu$ (see \Cref{proof:symmetric-matrix-dependence}), $A$ can be restricted to be symmetric positive definite without loss of generality. If the two covariance matrices are full rank, then the minimum is unique and achieved when

\begin{equation*}
M^{-1} = \spdmean{\cov{x}}{\cov{\alpha}^{-1}},
\end{equation*}
where $\spdmean{A}{B}$ is the geometric mean in the affine-invariant Riemannian metric (AIRM) of symmetric, positive-definite matrices.  

\begin{theorem}
Let $A$ be a solution to \Cref{eq:dense_equation} for $k$ samples and $k$ scores with a target distribution $\normal{\mu, \Sigma}$. Then there is a $k - 1$ dimensional subspace of $\vecspan{\drw{\x}{i} - \overline{\drw{\x}{i}}} \cup \vecspan{\drw{\alpha}{i} - \overline{\drw{\alpha}{i}}}$ in which $\Sigma$ and the covariance of $p_{\!\A}$ match precisely.

It follows that if $k > d + 1$, then $AA^\top = \Sigma$ and we recover the target distribution exactly.
\end{theorem}
\begin{proof}
See \Cref{appendix:mvnorm}.
\end{proof}

If the number of draws $k$ is smaller than $d$, \Cref{eq:dense_equation} has an infinite set of solutions. We add a regularization term $\gamma (\tr{AA^\top} + \tr{(AA^\top)^{-1}})$ to our minimization objective, which yields the optimality condition
\begin{equation}
\label{eq:dense_equation_reg}
    A \, A^\top \! (\cov{\alpha} + \gamma I) \, A \, A^\top = \cov{x} + \gamma I,
\end{equation}
which has the unique solution
\[
M_\gamma^{-1} = \spdmean{\left( \cov{x} + \gamma \Id\right)}
                   {\left(\cov{\alpha} + \gamma \Id\right)^{-1}}.
\]
 
\subsubsection{Low rank plus diagonal transformations}

Diagonal mass matrices are inexpensive because all operations are linear in dimension, but they cannot adjust for correlation at all.  Dense estimates can adjust for arbitrary correlations, but they quickly become infeasible as dimensionality increases, requiring $\mathcal{O}(d^3)$ time for estimation and $\mathcal{O}(d^2)$ per leapfrog step during Hamiltonian simulation, as well as requiring $\mathcal{O}(d^2)$ space.

In \Cref{appendix:min-lowrank} we show that the dense solution only depends on the subspace spanned by the centered scores and gradients. Outside of this subspace, the solution to the optimization problem is unconstrained. By adding a regularization term $R(A, \mu) = \gamma (\tr{AA^\top} + \tr{(AA^\top)^{-1}})$ to the optimization objective, we can ensure a unique solution such that $AA^\top$, constrained to that previously unconstrained space, is the identity matrix. We can represent and compute the resulting transformation efficiently in high dimensions as
\begin{equation}
F_1(y) = \left( \Id + U \, \left(\diag{\lambda}^{1/2} - \Id \right) \, U^\top\right) \, y + \mu_1,
\end{equation}
where $U \in \mathbb{R}^{d\times k}$ has orthonormal columns (i.e., $U^\top U = \Id$).
Each column of $U$ encodes a direction and $F_1$ scales $y$ in that direction by the corresponding value in $\lambda \in (0, \infty)^k$. All other orthogonal directions remain unchanged. See \Cref{alg:lowrank} for how to compute $U$ and $\lambda$.

The regularization term effectively "fills up" all unidentifiable eigenvalues with ones. There is however a priori no good reason to assume that missing eigenvalues would be one. For example, if the posterior is a high-dimensional $\normal{0, 10^{-5}\Id}$, we would use a very ill-conditioned transformation, because we are filling in eigenvalues with one instead of $10^{-5}$. But we previously showed that a diagonal transformation penalizes $\sum \lambda_i + \lambda_i^{-1}$, essentially "centering" the remaining eigenvalues at one. This suggests combining the diagonal and low-rank transformations.

We split the transformation $F$ into two parts such that $F = F_2 \circ F_1$,
\[
\begin{tikzcd}
\A_1 \arrow[r, "F_1"'] \arrow[rr, "F" description, bend left] & \A_2 \arrow[r, "F_2"'] & \T,
\end{tikzcd}
\]
where $F_1$ is a low-rank operation and $F_2$ operates only on the diagonal as described in \Cref{sec: diag_mass_matrix}.  Given a rank $k \leq d$ matrix $U \in \mathbb{R}^{d\times k}$ with orthonormal columns (i.e., $U^\top U = \Id$), 
\begin{eqnarray}
F_1(y) & = & \left( \Id + U \, \left(\diag{\lambda}^{1/2} - \Id \right) \, U^\top\right) \, y + \mu_1
\\[4pt]
F_2(v) & = & \sigma \odot v + \mu_2.
\end{eqnarray}

The combined transformation $F = F_2 \circ F_1$ corresponds to the mass matrix
\[
M = \diag{\sigma}^{-1}\left(\Id + U \, \left(\diag{\lambda}^{-1} - \Id \right) \, U^\top\right) \, \diag{\sigma}^{-1}.
\]

In practice, we avoid storing $M$ and work with $\sigma, \lambda$ and $U$ directly, which requires only $\bigo{kd}$ storage. During sampling, we only need products $M^{-1}\rho$ and evaluations of $F(y)$ to sample from $\normal{0, M}$, both of which can easily be achieved in $\bigo{kd}$ as well.

To further reduce the computational cost of leapfrog steps, we can omit entries in $\sigma$ that are close to one and don't affect the quality of the mass matrix much anyway.

\begin{algorithm}
  \caption{\textsc{low-rank-adapt}}
  \label{alg:lowrank}
  \begin{algorithmic}
    \Require $\x\in\mathbb{R}^{d\times n}: \textrm{draws}, \alpha\in\mathbb{R}^{d\times n}: \textrm{scores}, c: \textrm{cutoff},\;\gamma: \textrm{regularizer}$
    \vspace*{6pt}
    % \State $\sigma \in \mathbb{R}^d$
    % \\
    \For{$i$ in $0:d-1$}
        \State $\sigma_i \gets \sqrt{\var{\x_i} / \var{\alpha_i}}$
            \Comment{set coordinate-wise standard deviations}
    \EndFor
    \\
    \State $\x \gets (\x - \overline{\x}) \,\odot\,\sigma, \quad 
            \alpha \gets (\alpha - \overline{\alpha}) \,\odot\,\sigma$
      \Comment{apply diagonal transform}
    \\
    \State $U^\x \gets \textsc{svd}(\x),\quad U^\alpha \gets \textsc{svd}(\alpha)$
    \\
    \State $Q,\_ \gets \textsc{qr-thin}\bigl([\,U^\x\;\;U^\alpha\,]\bigr)$
      \Comment{get jointly-spanned orthonormal basis}
    \\
    \State $P^\x \gets Q^\top \x,\quad P^\alpha \gets Q^\top \alpha$
      \Comment{project onto shared subspace}
    \\
    \State $C^\x \gets P^\x (P^\x)^\top + \gamma I, \quad 
            C^\alpha \gets P^\alpha (P^\alpha)^\top + \gamma I$
      \Comment{empirical covariance, regularize}
    \\
    \State $\Sigma \gets \textsc{spdm}({C^\x}, {C^\alpha})$
      \Comment{solve $\Sigma\,C^\alpha\,\Sigma = C^\x$ for $\Sigma$}
    \State $U, \Lambda \gets \textsc{eigendecompose}(\Sigma)$
      \Comment{extract eigendecomposition}
    \State $U_c \gets \{\,U_i : \lambda_i \le 1/c \;\textrm{ or }\; \lambda_i \ge c\}$
    \Comment{filter eigenpairs}
    \State $\Lambda_c \gets \{\,\lambda_i : \lambda_i \le 1/c \;\textrm{ or }\; \lambda_i \ge c\}$
    \\
    \State \Return $Q, \, U_c, \, \Lambda_c$
  \end{algorithmic}
\end{algorithm}

\begin{algorithm}[t]
\caption{\textsc{symmetric positive-definite mean (spdm)}}
\label{alg:spdm}
\begin{algorithmic}
  \Require $A, B \in \mathbb{R}^{n\times n}$ symmetric, positive-definite \vspace*{6pt}
  \State $U_A, \Lambda_A \gets \textsc{eigendecompose}(A)$ \vspace*{3pt}
  \State $A^{{1/2}} \gets U_A \, \textrm{diag}({\Lambda_A}^{1/2}) \, U_A^\top$ \vspace*{3pt}
  \State $A^{-{1/2}} \gets U_A \, \textrm{diag}(\Lambda_A^{-1/2}) \, U_A^\top$ \vspace*{3pt}
    \State $M \gets A^{1/2} \, B \, A^{1/2}$ \vspace*{3pt}
    \State $U_M, \Lambda_M \gets \textsc{eigendecompose}(M)$ \vspace*{3pt}
    \State $M^{1/2} \gets U_M \, \textrm{diag}(\Lambda_M^{1/2}) \, U_M^\top$ \vspace*{6pt}
    \State \Return $A^{-1/2} \, M^{1/2} \, A^{-1/2}$
\end{algorithmic}
\end{algorithm}

\section{Adaptation schedule}

In order to turn our optimal update rules into a mass-matrix adaptation algorithm, we need to introduce a schedule dictating how often during sampling the mass matrix is updated and which set of draws and scores are used for the update.  While we follow roughly the standard pattern of adaptation, this section introduces two novel modifications to the algorithm used by Stan and other probabilistic programming languages.

Whether adapting a mass matrix using the posterior variance, as Stan does, or using a matrix based on the Fisher divergence, preconditioning invariably faces the same chicken-and-egg problem, namely that generating suitable posterior draws requires a good mass matrix, but estimating a good mass matrix requires acceptable posterior draws. 

Stan and other probabilistic programming languages begin with the identity transformation, collect a number of draws, then based on those draws, estimate a better transformation, and repeat \citep{betancourt_conceptual_2018, phan2019composable, cabezas2024blackjax}.  These systems, as well as \texttt{nutpie}, do not use standard HMC, but rather use the No-U-Turn Sampler (NUTS) \citep{hoffman2014no,betancourt_conceptual_2018}, where the length of the Hamiltonian trajectory is chosen dynamically.
%(For the full algorithm, see \Cref{appendix_nuts}.) 

With NUTS, it is extremely costly to generate draws with a poorly tuned mass matrix, because in ill-conditioned cases, the algorithm will require a minuscule step size for stability and a huge number of steps before satisfying the U-turn condition.  This problem is exacerbated by Stan's default cap of the number of leapfrog steps per draw of $2^{10}=1024$.  This is another way of saying that HMC is highly sensitive to condition number \citep{langmore_condition_2020}.  By default, Stan starts adaptation with a step-size adaptation window of 75 draws, using an identity mass matrix. This is followed by a mass matrix adaptation window consisting of a series of discrete intervals of increasing length, the first of which is 25 draws long. These initial 100 draws before the first adapted mass matrix can, perhaps counterintuitively, constitute a sizable fraction of the total sampling time over Stan's default 2000 iterations. 

Avoiding these costly warmup scenarios motivates the use of all available information about the posterior as quickly as possible during sampling. To this end, the \texttt{nutpie} sampler implements a faster-paced adaptation schedule than the NUTS algorithm.  We propose two fairly small modifications that provide a substantial boost in efficiency.  

\subsection{Improved initialization}

The first modification introduced by \texttt{nutpie} is that the mass matrix $M$ is initialized to $\diag{\left| \drw{\alpha}{0} \right|}$, where $\drw{\alpha}{0}$ is the score of the first draw.  The absolute score is the solution to \Cref{eq:dense_equation} when $\cov{\alpha}$ is taken to be the initial score squared (i.e., estimated covariance, assuming a mean of zero), and $\cov{\x}$ is taken to be the identity. This can also be viewed as a regularized diagonal of the gradient outer-product, which is a common approximation to the Hessian at $\x_0$, used in systems such as L-BFGS optimization \citep{nocedal2006numerical}. As a result, we have the following straightforward result.

\begin{lemma}[Scale free]
The \texttt{nutpie} adaptation algorithm is scale free in the sense that sampling from the unnormalized densities $p(\theta)$ and $q(\theta) = p(c \odot \theta)$ for any vector $c$ has the same efficiency.
\end{lemma}
\begin{proof}
The initialization at $M = \diag{\left| \drw{\alpha}{0} \right|}$ scales with $c$, as do all subsequent scores and hence trajectories.
\end{proof}

\subsection{More frequent updates}

The second modification introduced by \texttt{nutpie} is to update the mass matrix estimates more frequently, starting from the very first draw.  Stan waits until after 75 warmup iterations (by default) before it uses draws for estimating the mass matrix, in the hope that the Markov chain will have approximately converged to the stationary distribution; this is not unreasonable, as HMC mixes very quickly.  We do not want to wait that long.  We would rather integrate information from draws into estimates of the mass matrix immediately upon sampling them, while also discounting very early draws as sampling progresses.  We discount earlier draws because they will not be drawn from the (approximate) stationary distribution.  

For diagonal mass matrix estimates, we use Welford's algorithm to store online estimates of the variance of draws and scores to compute an up-to-date value for \Cref{eq:diag_solution} at each iteration. Instead of modifying the algorithm with a weighting scheme to taper the influence of these early draws, we periodically chop off the tail end from the set of draws our estimator is built upon. In other words, old draws are periodically wiped from the memory of the online estimate. The frequency with which this memory wipe occurs is a tunable parameter, $L$. With this approach, the number of draws informing the estimates after the first $L$ oscillates between $L$ and $2L$.  The matrix used to sample draw $n$ is always based on an estimate constructed from the sequence of draws $\drw{\theta}{a_n}, \ldots, \drw{\theta}{n-1}$, where
\begin{equation*}
    a_n = \max\bigl(0,\,L\,(\left\lfloor n / L \right\rfloor-1)\bigr).
\end{equation*}
The implementation of this scheme for diagonal \texttt{nutpie} is illustrated with an example of the first 20 draws in \Cref{appendix:foreground-background}. 

Low-rank plus diagonal adaptation requires us to store the draws and the scores in the window used for estimating the mass matrix.  In this case, the mass matrix is only updated at multiples of $L$ iterations. 

\subsection{Step-size adaptation phases}

The matrix adaptation occurs in conjunction with the dual averaging step-size adaptation of standard NUTS, which aims to tune $\epsilon$ to achieve a target acceptance rate over states in the last half of the Hamiltonian trajectories.  The stochastic gradient update on which dual averaging is based assumes a squared error, so gradient descent pushes step size lower if the acceptance rate is too low and higher if it is too high, lowering the learning rate over time.

We organize warmup into three phases based on the hyperparameters used to tune the mass matrix $M$ and step size $\epsilon$ adaptation.  
\begin{enumerate}
    \item The first 30\% of warmup iterations, adapts a mass matrix with an update frequency $L = 10$ using the Metropolis acceptance probability $\exp(\min(0,\Delta H))$ where $\Delta H$ is the change in the Hamiltonian induced by the error in the leapfrog integrator.
    \item The next 55\% of warmup iterations use $L = 80$ and the same acceptance probability for tuning. At the start of the second phase, the step size is reinitialized in the same manner as at the first draw.
    \item Following NUTS as implemented in Stan, the final 15\% of the warmup draws are devoted exclusively to fine-tuning the step size $\epsilon$.  Unlike the first two phases, the third phase uses the symmetric acceptance statistic $$\frac{2\exp(\min(0,\Delta H))}{1 +\exp(\Delta H)}.$$
\end{enumerate}

\section[Implementation and evaluation on posteriordb]{Implementation and evaluation on \texttt{posteriordb}}

We systematically test \texttt{nutpie}'s diagonal and low-rank samplers on the models in \texttt{posteriordb}, using Stan's default NUTS implementation as a benchmark. We run 4 chains of each sampler for 1000 warmup iterations and 1000 posterior draws on all models, using a target acceptance rate of 0.8. For \texttt{nutpie}'s low-rank mode, we use an eigenvalue cutoff $c = 2$, and regularization parameter $\gamma = 10^{-5}$. From our tests we exclude some exceptionally slow models, listed in \Cref{appendix:posteriordb}. We also discard models which did not converge for any of the solver configurations, defining convergence as if there are no divergences during the run, and the effective sample size exceeds 200 for all parameters. 114 models remained after this filtering. Among these, \texttt{nutpie}'s diagonal mode required a median \textasciitilde0.75x the number of gradient evaluations as Stan for an effective sample. In wall time, this translated to a \textasciitilde1.3x median speedup. The low-rank mode had a median \textasciitilde4x speedup in both effective sample size per gradient evaluation and ESS per second. Detailed results can be seen in \Cref{fig:ecdfs}. Additional diagnostic comparisons can be found in \Cref{appendix:posteriordb}.

\begin{figure}[H]
    \centering
    \includegraphics[width=0.95\textwidth]{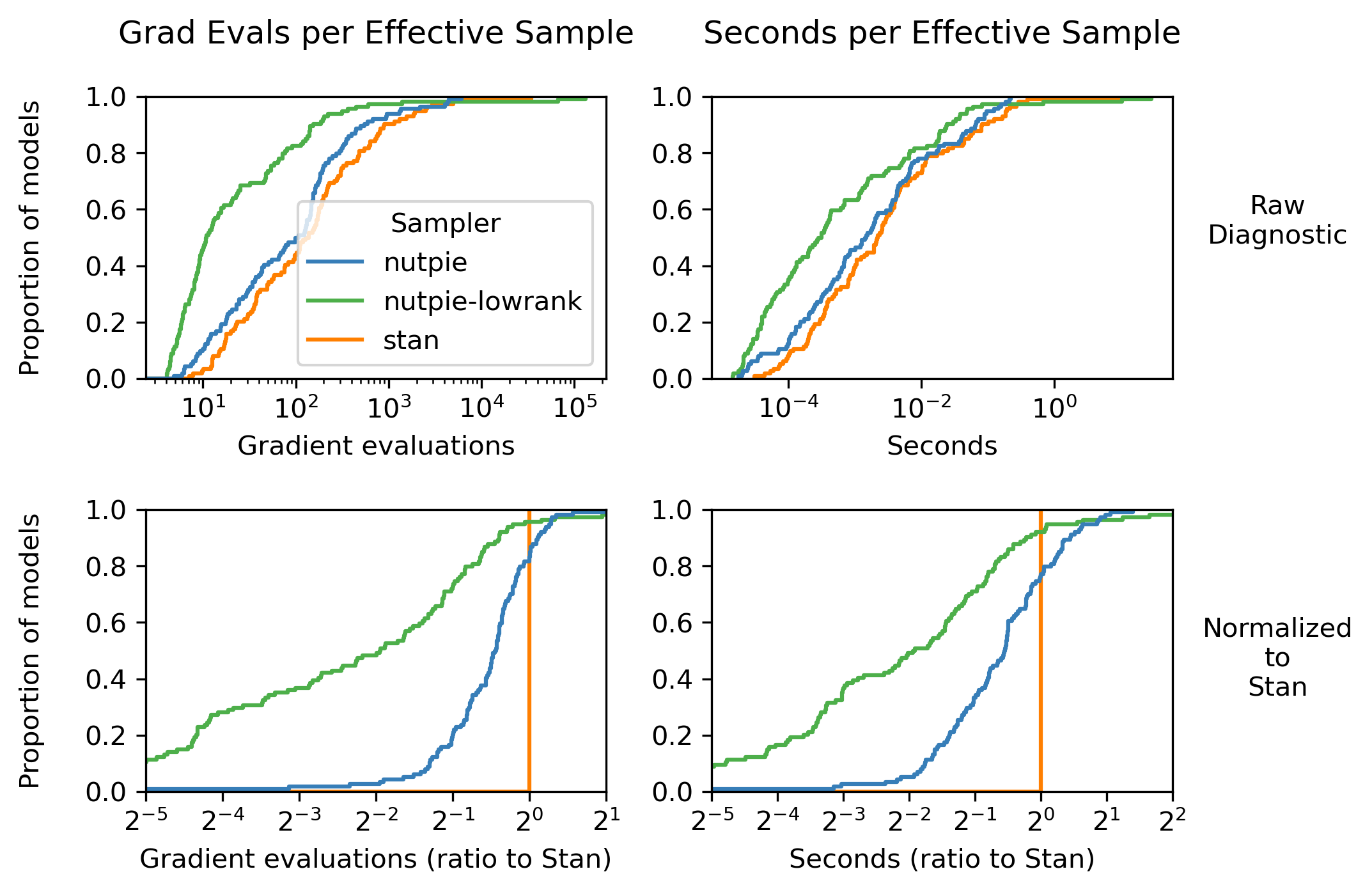}
    \caption{\textit{Cumulative distribution plots for target densities in \texttt{posteriordb}.  Each line represents a different sampler, \texttt{nutpie} (blue), Stan (orange), and \texttt{nutpie} low-rank (green) run for 1000 warmup iterations and 1000 posterior draws. The top row of plots show the raw diagnostic, while the bottom show the ratio of the diagnostic to Stan's. For instance, a point $(x,y)$ on the lower gradients plot says that a fraction $y$ of \texttt{posteriordb} models had $x$ or fewer times the number of gradient evaluations per effective sample as Stan's. Lines to the right of 1 show a small fraction of models for which Stan's default outperforms the other options.}}
    \label{fig:ecdfs}
\end{figure}

\begin{figure}
    \centering
    \includegraphics[width=0.95\textwidth]{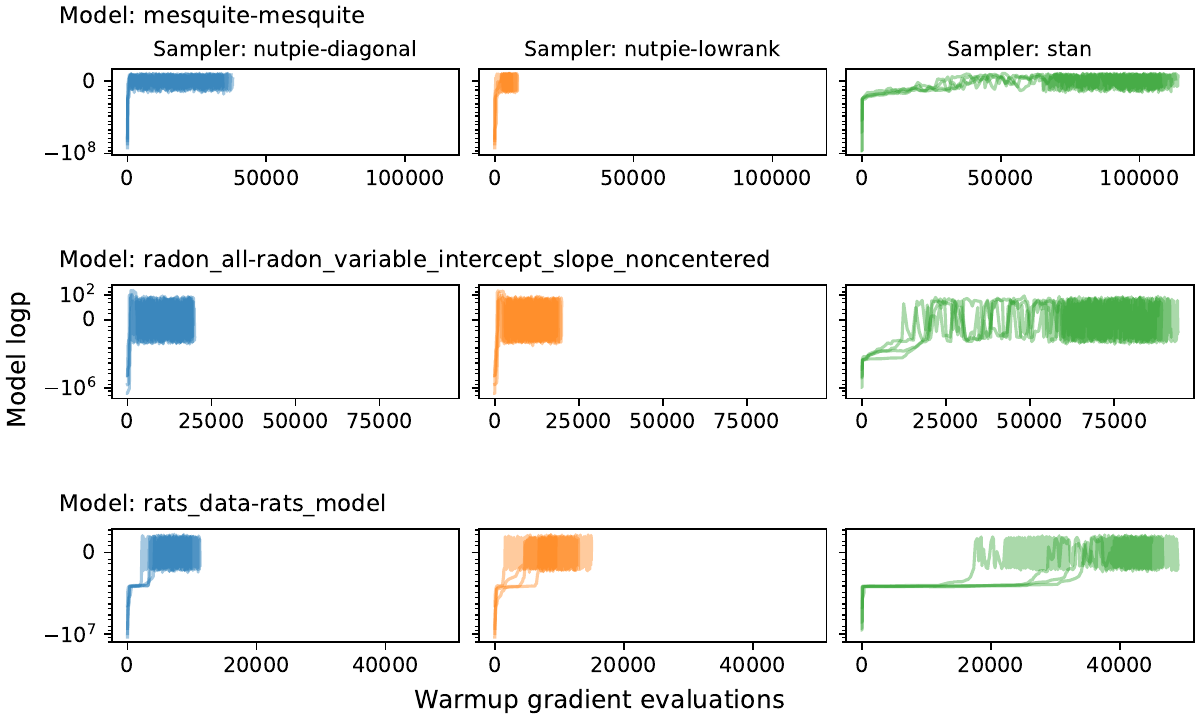}
    \caption{\textit{Trace-plot of sample log density scaled by number of gradient evaluations for three models from \texttt{posteriordb}. Each row corresponds to one model, the rows show \texttt{nutpie} with diagonal mass matrix adaptation, \texttt{nutpie} with low rank adaptation and Stan, respectively. The x-axis shows the number of gradient evaluations rather than the number of draws as typical with a trace-plot. All samplers run 1000 warmup draws. \texttt{nutpie} uses significantly fewer gradient evaluations in total, and avoids a long inefficient phase seen in Stan. \Cref{fig:warmup-trace-random} shows this trace-plot for 15 randomly selected models from \texttt{posteriordb}.}}
    \label{fig:warmup-trace}
\end{figure}

\section{More general transforms}

While in this paper we focus on affine transformations $F$, corresponding to a constant mass matrix in the HMC metric interpretation, the same framework applies to more general, nonlinear transformations that induce position-dependent metrics. This recovers a form of Riemannian HMC, where the metric is learned by minimizing the Fisher divergence. We have implemented an experimental version of this in \texttt{nutpie} that uses normalizing flows; see \Cref{appendix:transformed-hmc} for the geometric formulation.

\section{Open-source implementation}

All evaluations of Fisher-divergence based preconditioning in this paper were carried out using \texttt{nutpie}, an optimized Rust-based implementation that is distributed with a permissive open-source MIT License through the PyMC project.\footnote{See \url{https://pymc-devs.github.io/nutpie/}.}  Stan was used for the NUTS comparison.\footnote{See \url{https://mc-stan.org/docs/reference-manual/mcmc.html}.}

\section{Conclusion}

We have presented algorithms to construct three novel forms of linear preconditioner for Hamiltonian Monte Carlo which incorporate
information from the scores of warmup draws to better approximate target covariance. The added
information in the scores allows for quicker adaptation of the mass matrix than state of the art
HMC samplers. 

We have limited our discussion to linear transformations, but the Fisher
divergence framework naturally generalizes to nonlinear transformations and to
transformations on Riemannian Manifolds, yielding
position-dependent metrics and a practical link to Riemannian HMC.
Experiments with normalizing flows parameterized by neural networks,
as well as smaller model-informed transformations, are promising, and
we plan to report them separately.

\newpage
\bibliographystyle{apalike}
\bibliography{FisherHMCPaper}

@article{langmore_condition_2020,
  title = {A Condition Number for {Hamiltonian Monte Carlo}},
  author = {Langmore, Ian and Dikovsky, Michael and Geraedts, Scott and Norgaard, Peter and Von Behren, Rob},
  journal = {arXiv preprint arXiv:1905.09813},
  year = {2020}
}

@incollection{neal2011mcmc,
  title={{MCMC} using {Hamiltonian} dynamics},
  author={Neal, Radford M},
  booktitle={Handbook of {Markov} chain {Monte Carlo}},
  editor={Brooks, Steve and Gelman, Andrew and Jones, Galin and Meng, Xiao-Li},
  pages={47--95},
  year={2011},
  publisher={Chapman and Hall/CRC}
}

@article{tran_tuning_2024,
  title={Tuning diagonal scale matrices for {HMC}},
  author={Tran, Jimmy Huy and Kleppe, Tore Selland},
  journal={Statistics and Computing},
  volume={34},
  number={6},
  pages={196},
  year={2024},
  publisher={Springer}
}

@article{carpenter_stan_2017,
  title = {Stan: A Probabilistic Programming Language},
  author = {Carpenter, Bob and Gelman, Andrew and Hoffman, Matthew D. and Lee, Daniel and Goodrich, Ben and Betancourt, Michael and Riddell, Andrew},
  year = {2017},
  journal = {Journal of Statistical Software},
  volume = {76},
  number = {1},
  pages = {1--32},
  doi = {10.18637/jss.v076.i01}
}

@article{hoffman2014no,
  title={The {No-U-Turn} sampler: adaptively setting path lengths in {Hamiltonian Monte Carlo}},
  author={Hoffman, Matthew D and Gelman, Andrew},
  journal={Journal of Machine Learning Research},
  volume={15},
  number={1},
  pages={1593--1623},
  year={2014}
}

@misc{magnusson_posteriordb_2024,
  title={posteriordb: Testing, Benchmarking and Developing {B}ayesian Inference Algorithms},
  author={Måns Magnusson and Jakob Torgander and Paul-Christian Bürkner and Lu Zhang and Bob Carpenter and Aki Vehtari},
  booktitle={Proceedings of The 28th International Conference on Artificial Intelligence and Statistics},
  pages={1198--1206},
  year={2025},
  journal={Proceedings of Machine Learning Research},
  month={03--05 May},
  publisher={PMLR},
  url={https://proceedings.mlr.press/v258/magnusson25a.html},
}

@article{salvatier2016probabilistic,
  title={Probabilistic programming in {P}ython using {PyMC3}},
  author={Salvatier, John and Wiecki, Thomas V and Fonnesbeck, Christopher},
  journal={PeerJ Computer Science},
  volume={2},
  pages={e55},
  year={2016},
  publisher={PeerJ Inc.}
}

@misc{betancourt_conceptual_2018,
      title={A Conceptual Introduction to {H}amiltonian {M}onte {C}arlo}, 
      author={Michael Betancourt},
      year={2018},
      eprint={1701.02434},
      archivePrefix={arXiv},
      primaryClass={stat.ME},
      url={https://arxiv.org/abs/1701.02434},
      journal={arXiv preprint arXiv:1701.02434}
}

@article{welford1962,
    title={Note on a Method for Calculating Corrected Sums of Squares and Products},
    author={Welford, B.P.},
    year={1962},
    journal={Technometrics},
    volume={4},
    number={3},
    pages={419--420},
}

@article{cabezas2024blackjax,
  title={{BlackJAX}: Composable {{B}ayesian} inference in {JAX}},
  author={Cabezas, Alberto and Corenflos, Adrien and Lao, Junpeng and Louf, R{\'e}mi and Carnec, Antoine and Chaudhari, Kaustubh and Cohn-Gordon, Reuben and Coullon, Jeremie and Deng, Wei and Duffield, Sam and others},
  journal={arXiv preprint arXiv:2402.10797},
  year={2024}
}

@article{bales2019,
  title={Selecting the Metric in {Hamiltonian Monte Carlo}},
  author={Bales, Ben and Pourzanjani, Arya and Vehtari, Aki and Petzold, Linda},
  journal={arXiv preprint arXiv:1905.11916},
  year={2019}
}

@article{wainwright2008graphical,
  title={Graphical models, exponential families, and variational inference},
  author={Wainwright, Martin J and Jordan, Michael I},
  journal={Foundations and Trends in Machine Learning},
  volume={1},
  number={1--2},
  pages={1--305},
  year={2008},
  publisher={Now Publishers, Inc.}
}

@book{leimkuhler2004simulating,
  title={Simulating Hamiltonian Dynamics},
  author={Leimkuhler, Benedict and Reich, Sebastian},
  year={2004},
  publisher={Cambridge University Press}
}

@article{phan2019composable,
  title={Composable effects for flexible and accelerated probabilistic programming in {N}um{P}yro},
  author={Phan, Du and Pradhan, Neeraj and Jankowiak, Martin},
  journal={arXiv preprint arXiv:1912.11554},
  year={2019}
}

@inproceedings{ge2018turing,
  title={Turing: a language for flexible probabilistic inference},
  author={Ge, Hong and Xu, Kai and Ghahramani, Zoubin},
  booktitle={International Conference on Artificial Intelligence and Statistics},
  pages={1682--1690},
  year={2018},
  organization={PMLR}
}

@article{betancourt2015hamiltonian,
  title={Hamiltonian {M}onte {C}arlo for hierarchical models},
  author={Betancourt, Michael and Girolami, Mark},
  journal={Current Trends in {B}ayesian Methodology with Applications},
  volume={79},
  number={30},
  pages={2--4},
  year={2015}
}

@article{papaspiliopoulos2007general,
  title={A general framework for the parametrization of hierarchical models},
  author={Papaspiliopoulos, Omiros and Roberts, Gareth O and Sk{\"o}ld, Martin},
  journal={Statistical Science},
  pages={59--73},
  year={2007}
}

@article{ando2004geometric,
  title={Geometric means},
  author={Ando, Tsuyoshi and Li, Chi-Kwong and Mathias, Roy},
  journal={Linear Algebra and its Applications},
  volume={385},
  pages={305--334},
  year={2004}
}

@inproceedings{malago2015information,
  title={Information geometry of the {G}aussian distribution in view of stochastic optimization},
  author={Malag{\`o}, Luigi and Pistone, Giovanni},
  booktitle={Proceedings of the 2015 ACM Conference on Foundations of Genetic Algorithms XIII},
  pages={150--162},
  year={2015}
}

@book{nocedal2006numerical,
  title={Numerical Optimization},
  author={Nocedal, Jorge and Wright, Stephen J},
  year={2006},
  publisher={Springer}
}

\newpage
\appendix

\section{The geometry of symmetric, positive-definite matrices}

Our mass matrices need to be symmetric and positive definite in order to define a proper covariance over momentum vectors.  Therefore, we will be working in the smooth manifold of of $d\times d$ symmetric positive-definite matrices, which is conventionally denoted by $\mathcal{S}_{++}^d$.

\subsection{The affine-invariant Riemannian metric}\label{app:AIRM}

There is more than one metric $d(A, B)$ that can be defined over symmetric, positive-definite matrices $A, B \in \mathcal{S}_{++}^d$.  We will concentrate on the affine-invariant Riemannian metric (AIRM), which is defined as follows.

\begin{definition}[AIRM]
The \emph{affine--invariant Riemannian metric} on $\mathcal{S}_{++}^d$ is defined for $P\in \mathcal{S}_{++}^d$ and tangent vectors $U,V\in T_P\mathcal{S}_{++}^d \cong \mathbb{S}^d$ by the inner product on the tangent space at $P$,
\[
\innerproduct{P}{U}{V}
= \tr{P^{-1} U\, P^{-1} V}.
\]
\end{definition}

\begin{theorem}[Distance]
Distance in the affine-invariant Riemannian metric satisfies
$$
d(A, B) 
= \left| \left| \, \log\left(A^{-1/2} \, B \, A^{-1/2}\right) \, \right|\right|_F,
$$
where the norm $|| X ||_F = \sqrt{\tr{X^\top \, X}}$ is the Frobenius norm and the logarithm is the matrix logarithm.
\end{theorem}

\begin{corollary}[Affine invariance]
This distance function is affine-invariant in the sense that for any conforming invertible matrix $X$,
$$
d(A, B) = d(X^\top \! A \, X, \,
            X^\top \! B \, X).
$$
\end{corollary}

\begin{theorem}[Dense geometric mean]\label{thm:dense-geometric-mean}
Given $A,B \in \mathcal{S}_{++}^d$, their \emph{geometric mean} 
with respect to the affine-invariant Riemannian metric is the midpoint of the unique geodesic from $A$ to $B$, which is given by
\[
\spdmean{A}{B}
= A^{1/2} \, 
  \left(A^{-1/2} \, B \, A^{-1/2}\right)^{1/2}
  \, A^{1/2}.
\]
\end{theorem}
\begin{proof}
    See \citet{ando2004geometric}.
\end{proof}

\begin{corollary}[Diagonal geometric mean]\label{thm:diagonal-geometric-mean}
Given $a, b \in (0, \infty)^d$, let $A = \diag{a}$ and $B = \diag{b}$.  The geometric mean for diagonal matrices is
\[
\spdmean{\diag{a}}{\diag{b}}
= \diag{\sqrt{a \odot b}},
\]
where $\odot$ is elementwise multiplication and square root applies elementwise.
\end{corollary}
\begin{proof}
   Diagonal multiplication, square root, and inverse square root apply elementwise to diagonal elements. 
\end{proof}

Given a data generating density $p(x \mid \theta)$ with support over $x \in \mathbb{R}^d$ and parameters $\theta$, the Fisher information metric is defined at a fixed parameter value $\theta$ by
$$
\mathcal{I}(\theta)
= -\int_{\mathbb{R}^d} p(x \mid \theta)\,\frac{\partial^2}{\partial \theta \, \partial \theta^\top} \, \log p(x \mid \theta) \, \textrm{d} x.
$$

\begin{theorem}[AIRM \& Fisher info. for Gaussians]
The affine-invariant Riemannian metric is equal to the Fisher information metric for the family of centered multivariate Gaussian distributions $\normal{0,\Sigma}$.
\end{theorem}
\begin{proof}
See Equation~(23) of \citet{malago2015information}.
\end{proof}

\section{A geometric view of Fisher HMC}
\label{appendix:transformed-hmc}

This appendix presents a coordinate-free formulation of preconditioning for HMC,
extending the usual Euclidean treatment to general smooth manifolds and to
nonlinear (non-affine) transformations. Such a formulation also naturally encompasses settings where the target or latent space has non-Euclidean geometry (for example, distributions on spheres, Stiefel manifolds, or other Riemannian manifolds), although we focus on the Euclidean case in the main text.

Let the \emph{target space} be a smooth orientable manifold $\T$ endowed with a
posterior volume form $p$ (the target distribution viewed as an $n$-form with $\int p = 1$).

We augment $\T$ with a corresponding manifold of momenta to the cotangent bundle $T^*\T$. Any Riemannian metric $g$ on $\T$ with volume form $\vol_g$ allows us to define the Hamiltonian on the cotangent bundle $T^*\T$ of positions and momenta as
\[
H_{g,p}(x, \rho)
= \frac{1}{2}\norm{\rho}_{g^{-1}(x)}^2
  - \log\left(\frac{p}{\vol_g}\right)(x),
\]
where $p / \vol_g \in \Omega^0(\T)$ is the density of $p$ with respect to $\vol_g$.

Identify $T_{(x,\rho)}(T^*\T)\cong T_x\T \oplus T_x^*\T$.
Then the Hamiltonian vector field $X_H$ has components
\[
    X_H(x,\rho)=\left(\dd_\rho H_{g,p}(x,\rho),\; -\dd_x H_{g,p}(x,\rho)\right).
\]

Alternatively, $X_H$ is the vector field that satisfies $\iota_{X_H}\omega = \dd H$, where $\omega$ is the canonical symplectic form on $T^*\T$.

The vector field $X_H$ induces the flow $\Phi_H^t$ and yields the following (exact) HMC algorithm: starting at position $x_0 \in \T$, sample momentum $\rho_0 \in T^*_{x_0}\T \sim \normal{0, g_{x_0}}$. Evolve the system to time $t$ to obtain $(x_1, \rho_1) = \Phi_H^t(x_0, \rho_0)$, and return the new position $x_1$.

The performance of the sampler critically depends on the choice of metric $g$.  
Suppose that we already know that a distribution $q$ (viewed again as a volume form) on another manifold $\A$, diffeomorphic to $\T$, can be sampled efficiently with a metric $\mu$ using the Hamiltonian $H_{\mu,q}$. In practice, this is often $\A = \R^n$ with $\mu$ the constant identity metric and $q$ the standard normal volume form.

A diffeomorphism $f \colon \A \to \T$ pushes the metric forward to $\T$ through $f_*\mu$.
This allows us to sample from $p$ using the Hamiltonian $H_{f_*\mu, p}$ on $\T$, or equivalently, $H_{\mu, f^*p}$ on $\A$.

Since $H_{\mu,q}$ generates an efficient flow through the vector field $X_{H_{\mu, q}}$, we would like to choose $f$ so that the vector fields $X_{H_{\mu, q}}$ and $X_{H_{\mu, f^*p}}$ are similar.

% Pulling both Hamiltonians back to $\A$, this leads to the optimization problem
% \[
% \hat{f}
% = \argmin_{f \ \text{diffeo}}
% \int_\A\norm*{X_{H_{\mu, q}} - X_{H_{\mu, f^*p}}}^2_\mu f^*p.
% \]
Minimizing the squared error between the two Hamiltonian vector fields
on $T^{*}\mathcal{A}$ (evaluated at $\rho=0$) leads to the optimization
problem 
\[
\hat{f}=\argmin_{f\ \mathrm{diffeo}}\int_{\mathcal{A}}\left\Vert X_{H_{\mu,q}}(x,0)-X_{H_{\mu,f^{*}p}}(x,0)\right\Vert _{\mu}^{2}\,f^{*}p.
\]
To evaluate this, note that the Hamiltonian vector field has components
$X_H(x,\rho) = (\dd_\rho H, -\dd_x H)$.
For $H_{\mu,\omega}(x,\rho) = \frac{1}{2}\norm{\rho}^2_{\mu^{-1}} - \log(\omega/\vol_\mu)(x)$,
we have $\dd_\rho H_{\mu,\omega} = \mu^{-1}\rho$ and $\dd_x H_{\mu,\omega} = -\dd \log(\omega/\vol_\mu)$.
Thus, the position component of the vector field is $\mu^{-1}\rho$ for both Hamiltonians (independent of the target distribution), while the momentum components differ:
\[
X_{H_{\mu, q}} = \left(\mu^{-1}\rho, \dd\log\left(\frac{q}{\vol_\mu}\right)\right), \quad
X_{H_{\mu, f^*p}} = \left(\mu^{-1}\rho, \dd\log\left(\frac{f^*p}{\vol_\mu}\right)\right).
\]
The difference lies entirely in the momentum component. Moreover, it is
independent of $\rho$, so although the choice to evaluate the squared
error at $\rho=0$ may initially seem ad hoc, the resulting objective is
unchanged. Computing the squared norm with respect to $\mu$ (where the 
position part uses $\mu$ and the momentum part uses $\mu^{-1}$):
\[
\norm{X_{H_{\mu, q}} - X_{H_{\mu, f^*p}}}^2_\mu = \norm{0}^2_\mu + \norm{\dd\log(q/f^*p)}^2_{\mu^{-1}} = \norm{\dd\log(f^*p/q)}^2_{\mu^{-1}}.
\]
Therefore,
\[
\hat{f}
= \argmin_{f \ \text{diffeo}}
\int\norm*{\dd \log\left(\tfrac{f^*p}{q}\right)}^2_{\mu^{-1}} \, f^*p,
\]
where $\|\cdot\|_{\mu^{-1}}$ denotes the norm on $1$-forms induced by $\mu$.

The resulting cost function is exactly the Fisher divergence from $f^*p$ to $q$:

\begin{definition}[Fisher Divergence]
Let $X$ be a manifold with metric $g$ and volume forms $\omega_1,\omega_2$ that are mutually absolutely continuous.  
The \emph{Fisher divergence} from $\omega_1$ to $\omega_2$ is
\[
\mathcal{F}_{g}(\omega_1,\omega_2)
\;=\;
\int_{X} \norm*{\dd\log\left(\frac{\omega_1}{\omega_2}\right)}_{g^{-1}}^{2}\omega_1.
\]
\end{definition}

\subsection{Leapfrog steps}

In practice, we usually cannot compute $\Phi_{H}^t$ directly. But if the Hamiltonian can be written as the sum of two functions $H = V + K$ such that $\Phi_{V}^t$ and $\Phi_{K}^t$ can be computed efficiently, we can use the leapfrog integrator to approximate 
\[
\Phi_{H}^\epsilon \approx \Phi_V^{\epsilon/2}\circ \Phi_{K}^\epsilon\circ\Phi_{V}^{\epsilon/2}.
\]
Different splits of $H$ are possible, but generally we choose $K = \inner{\rho}{\rho}_{g^{-1}} / 2$ and $V = - \log\left(\frac{p}{\vol_g}\right)$, such that $\Phi_K^t$ is the exponential map with respect to $g$, and $V$ depends only on $x$ and not on $\rho$, so that $\Phi_{V}^t$ is also easy to compute\footnote{Another possible split that has some nice properties in this context is $K = \inner{\rho}{\rho}_{g^{-1}} / 2 - \log(q/\vol_g)$, the Hamiltonian of the reference distribution $q$, and $V = \log(p/q)$. The first flow can be easily solved analytically if $q$ is the standard normal distribution. This way, the integration error will go to zero as the Fisher divergence between q and p goes to zero.}. \Cref{alg:transformed-leapfrog} shows how we can use this to compute approximate $\Phi_H^t$ in a reversible manner.

\begin{algorithm}
  \caption{\textsc{transformed-leapfrog}}
  \label{alg:transformed-leapfrog}
  \begin{algorithmic}
    \Require $
        \x: \text{draw},
        v: \text{velocity},
        L: \text{num. steps},
        \epsilon: \text{step size},
        f: \text{diffeomorphism}
    $
    \State $\x^{(0)} \gets \x,\quad v^{(0)} \gets v$
    \State $y \gets f^{-1}\left(\x^{(0)}\right)$
    \State $\delta \gets \ddy \log\left(\tfrac{f^*p}{\lambda_\A}\right)(y)$
        \Comment{evaluate score on $\A$}
    \For{$i = 1$ \textbf{to} $L$}
      \State $v^{(i+1/2)} \gets v^{(i)} - \tfrac{\epsilon}{2}\,\delta$
        \Comment{half-step velocity}
      \State $y \gets y + \epsilon\,v^{(i+1/2)}$
        \Comment{full-step position (mass matrix = $\Id$)}
      \State $\delta \gets \ddy \log\left(\tfrac{f^*p}{\lambda_\A}\right)(y)$
      \State $v^{(i+1)} \gets v^{(i+1/2)} - \tfrac{\epsilon}{2}\delta$
        \Comment{half-step velocity}
    \EndFor
    \State \Return $f(\drw{y}{L})$
  \end{algorithmic}
\end{algorithm}

\newpage

\section{Proofs}

\subsection{Minimization of Fisher divergence for affine transforms}
\label{proof:fisher_divergence_min}

\begin{theorem}
Let $\drw{x}{i}$ be draws from a distribution with density $p$ and $F_{A, \mu}\colon \A\to\T$ with $F_{A, \mu}(\y) = A\y - \mu$ a family of diffeomorphisms. Then the sample Fisher divergence
\[
    \fdivhat{p_\A(y; F_{A, \mu}), \normal{0, \Id}} = \frac{1}{n}\sum_{i} \norm{\drw{y}{i} + \drw{\beta}{i}}^2
\]
is minimal if
\[
    M\,\cov{\drw{\x}{i}}\,M = \cov{\drw{\alpha}{i}},\qquad\mu=\overline \x +M\overline\alpha,    
\]
where $M = (AA^\top)^{-1}$ and
\begin{align}
    \drw{\y}{i} &= F^{-1}(\drw{\x}{i}) = A^{-1}(\drw{x}{i} - \mu) \\
    \drw{\beta}{i} &= J_F^\top(\drw{\y}{i}) \, \alpha^{(i)} + \ddy \log \absdet{J_F(\drw{\y}{i})} = A^\top \drw{\alpha}{i}.
\end{align}
\end{theorem}

\begin{proof}
Let $X=[\drw{x}{1},\dots,\drw{\x}{n}] \in R^{d\times n}$ and $S = [\drw{\alpha}{1},\dots,\drw{\alpha}{n}]$. Let $\ones=(1,\dots,1)^T\in\mathbb{R}^n$, then
\begin{align}
    D_\text{F} &= \frac{1}{n}\norm{A^\top S + A^{-1}(X-\mu \ones^\top)}^2_F,
\end{align}
where $\norm{\cdot}_F$ is the Frobenius norm.

\proofpart{Optimize with respect to $\mu$}
We get
\begin{align}
    \dd_\mu D_\text{F} &\propto
    \tr{(A^\top S + A^{-1}(X-\mu \ones^\top)^\top\dd_\mu (A^\top S + A^{-1}(X-\mu \ones^\top)}\\
    &= \tr{
        -(S^\top A + (X^\top - \ones\mu^\top) A^{-\top})A^{-1}\dd\mu\ones^\top
    }\\
    &= \tr{
        (\ones^\top\ones\mu^\top M - \ones^\top S^\top - \ones^\top X^\top M)\dd\mu
    }.
\end{align}

This is equal the constant zero function iff
\begin{gather}
(\ones^\top\ones\mu^\top M - \ones^\top S^\top - \ones^\top X^\top M) = 0 \\
\shortintertext{or}
\mu = \frac{1}{n}(X\ones + M^{-1} S\ones) = \overline{x} + M^{-1}\overline{\alpha},
\end{gather}
where $\overline{x}$ and $\overline{\alpha}$ are the sample means of $\drw{x}{i}$ and $\drw{\alpha}{i}$ respectively.

\proofpart{Optimizing with respect to $A$}
Write $\tilde X = X-\overline \x\ones^\top$ and $\tilde S = S-\overline\alpha\,\ones^\top$ for the centered samples and scores.  Then
\begin{align*}
\dd_A \fdivhatOp
&\propto
  \tr{
    \bigl(A^\top S + A^{-1}(X-\mu\ones^\top)\bigr)^\top
    \bigl(\dd A^\top S - A^{-1}\dd AA^{-1}(X-\mu\ones^\top)\bigr)
    }\\
&= 
    \begin{multlined}[t]
    \tr{
        (A^\top S + A^{-1}(\tilde X - M^{-1}\overline \alpha\ones^\top))S^\top\dd A
    }\\
    + \tr{
        A^{-1}(M^{-1}\overline\alpha\,\ones^\top - \tilde X)\,
        (A^\top \alpha + A^{-1}(\tilde X - M\overline\alpha\ones^\top))^\top A^{-1} \dd A
    }.
    \end{multlined}
\end{align*}
Setting $\dd \fdivhatOp=0$ for all $\dd A$ yields the matrix equation
\begin{multline}
0 
= \bigl(A^\top S + A^{-1}(\tilde X - M^{-1}\overline\alpha\ones^\top)\bigr)S^\top\\
+A^{-1}(M^{-1}\overline\alpha\ones^\top - \tilde X)
  \bigl(A^\top S + A^{-1}(\tilde X - M^{-1}\overline\alpha\ones^\top)\bigr)^\top A^{-1}.
\end{multline}

Using $\ones^\top\tilde X^\top=\ones^\top\tilde S^\top=0$ and rearranging gives
\[
0
= A^\top\tilde S\tilde S^\top + A^{-1}\tilde XS^\top
-A^{-1}\tilde X \tilde S^\top
-A^{-1}\tilde X \tilde X^\top\,M,
\]
which simplifies to the condition
\[
\tilde S\tilde S^\top =M\tilde X\tilde X^\top M
\]
\end{proof}

\subsection{Dependence on \texorpdfstring{$AA^\top$}{A times A transposed}}
\begin{theorem}
\label{proof:symmetric-matrix-dependence}
    Let
\[
\fdiv{p_\Phi}{q} = \frac{1}{N} \sum_{i=1}^N \left\| A^\top \alpha_i + A^{-1}(\x_i - \mu) \right\|^2.
\]
Then \( \fdiv{p_\Phi}{q} \) depends only on the symmetric matrix \( AA^\top \).
\end{theorem}

\begin{proof}
Expanding the squared norm gives
\[
\begin{aligned}
\left\| A^\top \alpha_i + A^{-1} (\x_i - \mu) \right\|^2
&= \left( A^\top \alpha_i + A^{-1} (\x_i - \mu) \right)^\top \left( A^\top \alpha_i + A^{-1} (\x_i - \mu) \right) \\
&= \alpha_i^\top A A^\top \alpha_i
+ 2 \alpha_i^\top A A^{-1} (\x_i - \mu)
+ (\x_i - \mu)^\top (A^{-1})^\top A^{-1} (\x_i - \mu)
\end{aligned}
\]
Note that $A A^{-1} = I_n \quad \text{and} \quad (A^{-1})^\top A^{-1} = (AA^\top)^{-1}$. Therefore the expression depends on $A$ only through the value of $AA^\top$.
\end{proof}

\subsection{Minimization of KL divergence for affine transforms}
\label{proof:kl-divergence-min}

\begin{theorem}
Let $\drw{x}{i}$ be draws from the distribution with density $p_\T$, and let $F_{\mu, A}\colon \A\to\T$ be the affine transformations $F_{\mu, A}(x)= Ax+\mu$. Let $p_\A$ be the density of the distribution of $p_\T$ pulled back to $\A$.

Then
\[
\argmin \kldivhat{p_\A, \normal{0, \Id}} = (\overline{x}, \cov{\drw{x}{i}}),
\]
where $\overline{x} = \frac{1}{n}\sum\drw{x}{i}$ and $\cov{\drw{x}{i}}$ is the maximum likelihood covariance estimate.
\end{theorem}

\begin{proof}
Because the KL divergence is parametrization independent, we get
\[
\kldivhat{p_\A, \normal{0, \Id}} = \kldivhat{p_\T, \normal{\mu, \Sigma}},
\]
where $\Sigma = AA^\top$.

\[
\kldivhat{p_\T, \normal{\mu, \Sigma}} = \sum_{i=1}^n \log p(\drw{x}{i}) + \frac{1}{2} \log \det\Sigma + \frac{1}{2} (\drw{x}{i} - \mu)^\top \Sigma^{-1}(\drw{x}{i} - \mu)
\]

\proofpart{Minimize with respect to $\mu$}

\[
\dd_\mu\kldivhat{p_\T, \normal{\mu, \Sigma}} = \sum_{i=1}^n (\drw{x}{i} - \mu)^\top \Sigma^{-1}\dd \mu
\]
is constant zero if
\[
\sum (\drw{x}{i} - \mu)^\top \Sigma^{-1} = 0,
\]
or
\[
\mu = \frac{1}{n}\sum \drw{x}{i} = \overline{x}.
\]

\proofpart{Minimize with respect to $\Sigma$}

We compute the differential of the KL divergence with respect to $\Sigma$:
\[
\dd_\Sigma \kldivhat{p_\T, \normal{\mu, \Sigma}} = \frac{1}{2}\sum_{i=1}^n \tr{\Sigma^{-1}\dd \Sigma} - (\drw{x}{i} - \mu)^\top \Sigma^{-1}\dd \Sigma \Sigma^{-1} (\drw{x}{i} - \mu).
\]
Setting $\drw{\tilde{x}}{i} = \drw{x}{i} - \mu = \drw{x}{i} - \overline{x}$, and using the cyclic trace property, we get
\[
\dd_\Sigma \kldivhat{p_\T, \normal{\mu, \Sigma}} = \frac{1}{2}\sum \tr{(\Sigma^{-1} - \Sigma^{-1}\drw{\tilde{x}}{i} \left({\drw{\tilde{x}}{i}}\right)^{\top} \Sigma^{-1})\dd \Sigma}.
\]
This is constant zero if
\[
\sum \Sigma^{-1} - \Sigma^{-1}\drw{\tilde{x}}{i} \left( {\drw{\tilde{x}}{i}} \right)^\top\Sigma^{-1} = 0,
\]
or if
\[
\Sigma = \frac{1}{n}\sum_i \drw{\tilde{x}}{i} \left( {\drw{\tilde{x}}{i}} \right)^\top = \cov{\drw{x}{i}}
\]

\end{proof}

\subsection{Exact recovery for multivariate normal distributions}
\label{appendix:mvnorm}

Let $X, S\in \R^{d\times k}$ be the matrix of centered draws and corresponding scores of $\normal{\mu, \Sigma}$. Let $A$ be a solution to $AA^\top SS^\top AA^\top = CC^\top$. We will show that $AA^\top = \Sigma$ restricted to the column space of $C$ or that $Mv = \Sigma v$ for all $v\in \text{col}(C)$.

\begin{proof}
A little algebra shows that $S = -\Sigma^{-1}C$. Set $N = \Sigma^{-1/2}AA^\top\Sigma^{-1/2}$ and $P = \Sigma^{-1/2}CC^\top\Sigma^{-1/2}$. Then our equation becomes

\[
NPN = P.
\]

Let
\[
P = U D' U^\top =
\begin{bmatrix}Q & Q_\perp\end{bmatrix}
\begin{bmatrix}D & 0 \\ 0 & 0\end{bmatrix}
\begin{bmatrix}Q^\top \\ Q_\perp^\top\end{bmatrix}
\]
be the singular value decomposition of $P$, such that $D \succ 0$, and write

\[
U N U^\top
= \begin{bmatrix} A & B \\ B^\top & E\end{bmatrix}.
\]

We get 

\[
\begin{bmatrix} A & B \\ B^\top & E\end{bmatrix}
\begin{bmatrix}D & 0 \\ 0 & 0\end{bmatrix}
\begin{bmatrix} A & B \\ B^\top & E\end{bmatrix}
=
\begin{bmatrix}
ADA & ADB \\
B^\top D A & B^\top D B
\end{bmatrix} = \begin{bmatrix}
D & 0 \\ 0 & 0 
\end{bmatrix}.
\]

The lower right block gives us $B^\top D B = 0$, and because $D \succ 0$, this implies $B = 0$.
The upper left block gives us $ADA = D$. But because $D\succ 0$ and $A \succeq 0$ this implies $A = \Id$.

Therefore
\[
UNU^\top
= \begin{bmatrix}
\Id & 0 \\
0 & E
\end{bmatrix},
\]

so $N$ acts like the identity on $\text{col}(Q) = \text{col}(P)$.

Now, let $v\in \text{col}(C)$. Since $\text{col}(P) = \text{col}(\Sigma^{-1/2}C)$, we have $\Sigma^{-1/2}v \in \text{col}(P)$. Therefore, $N\Sigma^{-1/2}v = \Sigma^{-1/2}v$. Multiplying by $\Sigma^{1/2}$ gives the desired $AA^\top v = \Sigma v$.
\end{proof}

\subsection{Minimum of low-rank transformation}
\label{appendix:min-lowrank}

Let $C \in \mathbb{R}^{d \times k}$ be the matrix of centered draws, $S \in \mathbb{R}^{d \times k}$ be the matrix of centered scores, and $M \in \mathbb{R}^{d \times d}$ be the symmetric, positive-definite mass matrix. 

Given a regularization parameter $\gamma \in \mathbb{R}^+$, the transformation is optimal if it satisfies the condition:
$$M (\gamma I + CC^\top) M = \gamma I + SS^\top$$
We define an orthogonal matrix $Q = [Q_{2k}, Q_r]$ such that $Q Q^\top = I$, where $Q_{2k} \in \mathbb{R}^{d \times 2k}$ spans the joint column space of $C$ and $S$. By construction, the orthogonal complement $Q_r$ satisfies
$$Q_r^\top C = 0 \quad \text{and} \quad Q_r^\top S = 0$$
To find the optimal $M$, we project the optimality condition onto the basis $Q$ by pre-multiplying by $Q^\top$ and post-multiplying by $Q$:
$$Q^\top M (\gamma I + CC^\top) M Q = Q^\top (\gamma I + SS^\top) Q$$
Because $Q Q^\top = I$, we can insert the identity between the terms on the left side:
$$(Q^\top M Q) \big( \gamma I + Q^\top CC^\top Q \big) (Q^\top M Q) = \gamma I + Q^\top SS^\top Q$$
Let $\tilde{M} = Q^\top M Q$ be the mass matrix in the projected space. We partition $\tilde{M}$ into blocks corresponding to $Q_{2k}$ and $Q_r$:
$$\tilde{M} = \begin{bmatrix} M_{2k} & M_{12} \\ M_{12}^\top & M_r \end{bmatrix}$$
Because $Q_r$ is orthogonal to $C$ and $S$, the projected covariance matrices are strictly block-diagonal. Letting $C_{2k} = Q_{2k}^\top C$ and $S_{2k} = Q_{2k}^\top S$, we have:
$$\begin{bmatrix} M_{2k} & M_{12} \\ M_{12}^\top & M_r \end{bmatrix} \begin{bmatrix} \gamma I + C_{2k}C_{2k}^\top & 0 \\ 0 & \gamma I \end{bmatrix} \begin{bmatrix} M_{2k} & M_{12} \\ M_{12}^\top & M_r \end{bmatrix} = \begin{bmatrix} \gamma I + S_{2k}S_{2k}^\top & 0 \\ 0 & \gamma I \end{bmatrix}$$
For brevity, let $D_C = \gamma I + C_{2k}C_{2k}^\top$ and $D_S = \gamma I + S_{2k}S_{2k}^\top$. Expanding the left side via block-matrix multiplication yields:
$$\begin{bmatrix} M_{2k} D_C M_{2k} + \gamma M_{12} M_{12}^\top & M_{2k} D_C M_{12} + \gamma M_{12} M_r \\ M_{12}^\top D_C M_{2k} + \gamma M_r M_{12}^\top & M_{12}^\top D_C M_{12} + \gamma M_r^2 \end{bmatrix} = \begin{bmatrix} D_S & 0 \\ 0 & \gamma I \end{bmatrix}$$
Equating the corresponding blocks gives us a system of equations. From the top-right block, we have:
$$M_{2k} D_C M_{12} + \gamma M_{12} M_r = 0$$
Because $M$ is symmetric and positive-definite, its principal submatrices $M_{2k}$ and $M_r$ are also symmetric and positive-definite. Furthermore, $D_C$ is strictly positive-definite ($\gamma > 0$). Therefore, the product $M_{2k} D_C$ has strictly positive eigenvalues. The equation above is a homogeneous continuous Sylvester equation of the form $A X + X B = 0$, where $A = M_{2k} D_C$, $X = M_{12}$, and $B = \gamma M_r$. Because $A$ and $B$ both have strictly positive eigenvalues, their spectra are disjoint, meaning the unique solution to this Sylvester equation is the zero matrix:
$$M_{12} = 0$$
Substituting $M_{12} = 0$ into the bottom-right block equation yields:
$$\gamma M_r^2 = \gamma I \implies M_r^2 = I$$
Since $M_r$ must be positive-definite, the unique solution is $M_r = I$. This proves that the optimal transformation acts purely as the identity on the orthogonal complement $Q_r$. Finally, substituting $M_{12} = 0$ into the top-left block gives the reduced optimality condition:
$$M_{2k} (\gamma I + C_{2k}C_{2k}^\top) M_{2k} = \gamma I + S_{2k}S_{2k}^\top$$
Solving the $\mathcal{O}(d^3)$ dense optimality condition is mathematically equivalent to solving the exact same problem in the $\mathcal{O}(k^3)$ projected subspace and padding the remaining dimensions with the identity matrix.

\subsection{Minimum of diagonal transformation}
\label{appendix:diag-proof}
\begin{proof}
Let $y = D^{-1}(x - \mu)$ be the affine transformation where $D = \text{diag}(\sigma)$. The empirical Fisher divergence between the transformed distribution and the standard normal $\mathcal{N}(0, I)$ is
$$\widehat{D}(\mu, D) = \frac{1}{n} \sum_{i=1}^n \left\| D \alpha_i + D^{-1}(x_i - \mu) \right\|^2.$$
Setting the gradient with respect to the shift $\mu$ to zero yields
$$\nabla_\mu \widehat{D} = -2 D^{-1} \left( D \overline{\alpha} + D^{-1}\overline{x} - D^{-1}\mu \right) = 0 \implies \mu^* = \overline{x} + D^2 \overline{\alpha}.$$
Because $D^2 = \text{diag}(\sigma^2)$, this is equivalent to $\mu^* = \overline{\x} + {\sigma^2}^* \odot \overline{\alpha}$. Substituting $\mu^*$ into the objective centers the draws and scores, $\tilde{x}_i = x_i - \overline{x}$ and $\tilde{\alpha}_i = \alpha_i - \overline{\alpha}$, decoupling the objective across the $d$ dimensions. For each dimension $j$, letting $v_j = \sigma_j^2$, we minimize the quantity $$v_j \text{var}(\alpha_j) + v_j^{-1} \text{var}(x_j) + 2\text{cov}(x_j, \alpha_j).$$
Setting the derivative with respect to $v_j$ to zero gives
$$\text{var}(\alpha_j) - v_j^{-2} \text{var}(x_j) = 0 \implies v_j = \sqrt{\frac{\text{var}(x_j)}{\text{var}(\alpha_j)}}.$$
The affine-invariant geometric mean of two positive scalars is $a \,\#\, b = \sqrt{ab}$. Recognizing that the marginal variances are exactly the diagonals of the respective empirical covariance matrices, rewriting the element-wise optimum $v_j$ in vector form directly yields
$${\sigma^2}^* = \invdiag{\cov{\x}} \,\#\ \left( \invdiag{\cov{\alpha}}\right)^{-1}.$$
\end{proof}
\newpage

\section{Adaptation online mean and variance algorithm}

\subsection{Welford's algorithm}
\label{appendix:welford}

Efficient, numerically stable, single-pass estimates of the sample mean and
(unbiased) sample variance via Welford’s algorithm.

\begin{algorithm}[H]
  \caption{\textsc{welford-update}}
  \begin{algorithmic}
  \Require $n: \textrm{iteration},\;\overline{x}: \textrm{running mean}, x: \textrm{new draw},\; M_2: \textrm{running sum sq. diff.}$
      \State $n \gets n + 1$
      \State $\delta \gets x - \overline{x}$
      \State $\overline{x} \gets \overline{x} + \delta / n$
      \State $M_2 \gets M_2 + \delta \cdot (x - \overline{x})$
      \State \Return $\bigl(\overline{x}, \; M_2/(n-1)\bigr)$ \Comment{Mean and unbiased variance}
  \end{algorithmic}
\end{algorithm}

Let $\overline{x}_n$ be the mean after $n$ observations and $M_{2,n}$ the
corresponding accumulated sum of squared deviations.  When $x_{n+1}$ arrives,
\begin{align*}
\overline{x}_{n+1} &= \overline{x}_n + \frac{x_{n+1}-\overline{x}_n}{n+1}, \\[4pt]
M_{2,n+1} &= M_{2,n} + (x_{n+1}-\overline{x}_n)(x_{n+1}-\overline{x}_{n+1}).
\end{align*}

\subsection{Foreground-background modification}
\label{appendix:foreground-background}
The removal of old draws from the Welford online estimator of draw and gradient variances is done by maintaining two parallel estimates, a ``foreground'' and ``background'', which are stagnated to preserve the estimate when old draws are wiped. The foreground estimate is what the sampler reads from at every iteration. The background estimator, while using the same update logic, maintains a fresher estimate based on a shorter window of recent draws, and periodically hands off its fresher estimate to the foreground before resetting. The foreground estimator then builds on this estimate until being replaced by a new one from the background. A diagram of this scheme is shown in \Cref{fig:adapt_fig}.

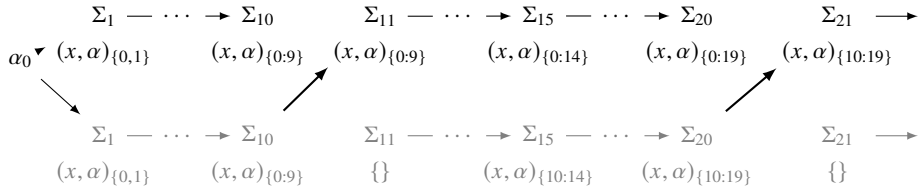
\begin{figure}[H]
\centering
\begin{tikzpicture}[
    scale=0.8,
    transform shape,
    every node/.style={
      draw,
      rectangle,
      align=center
    },
    >=latex,
  ]
  % bottom parameter
  \node[draw=none] (x) {$\alpha_0$};

  % background column
  \node (B1) [draw=none, text=gray, label={[text=gray]below:{$(\x, \alpha)_{\{0,1\}}$}}, below right=10mm of x] {$\Sigma_1$};
  \node[draw=none, text=gray] (D1) [right=5mm of B1] {$\dots$};
  \node (B2) [draw=none, text=gray, label={[text=gray]below:{$(\x, \alpha)_{\{0:9\}}$}}, right=5mm of D1] {$\Sigma_{10}$};
  \node (B3) [draw=none, text=gray, label={[text=gray]below:{$\{\}$}}, right=12mm of B2] {$\Sigma_{11}$};
  \node[draw=none, text=gray,] (D2) [right=5mm of B3] {$\dots$};
  \node (B4) [draw=none, text=gray, label={[text=gray]below:{$(\x, \alpha)_{\{10:14\}}$}}, right=5mm of D2] {$\Sigma_{15}$};
  \node[draw=none, text=gray] (D3) [right=5mm of B4] {$\dots$};
  \node (B5) [draw=none, text=gray, label={[text=gray]below:{$(\x, \alpha)_{\{10:19\}}$}}, right=5mm of D3] {$\Sigma_{20}$};
  \node (B6) [draw=none, text=gray, label={[text=gray]below:{${\{\}}$}}, right=15mm of B5] {$\Sigma_{21}$};
  \node[draw=none] (B7) [right=of B6]   {};

  % foreground column
  \node (F1) [draw=none, label=below:{$(\x, \alpha)_{\{0,1\}}$}, above=14mm of B1]  {$\Sigma_1$};
  \node (DD1) [draw=none, right=5mm of F1] {$\dots$};
  \node (F2) [draw=none, label=below:{$(\x, \alpha)_{\{0:9\}}$}, above=14mm of B2]  {$\Sigma_{10}$};
  \node (F3) [draw=none, label=below:{$(\x, \alpha)_{\{0:9\}}$}, above=14mm of B3]  {$\Sigma_{11}$};
  \node[draw=none] (DD2) [right=5mm of F3] {$\dots$};
  \node (F4) [draw=none, label=below:{$(\x, \alpha)_{\{0:14\}}$}, above=14mm of B4]  {$\Sigma_{15}$};
  \node (DD3) [draw=none,right=5mm of F4]  {$\dots$};
  \node (F5) [draw=none, label=below:{$(\x, \alpha)_{\{0:19\}}$}, above=14mm of B5]   {$\Sigma_{20}$};
  \node (F6) [draw=none, label=below:{$(\x, \alpha)_{\{10:19\}}$}, above=14mm of B6] {$\Sigma_{21}$};
  \node (F7) [draw=none, right=of F6]   {};
 
  \draw[->, shorten <=2pt, shorten >=2pt] (x) -- (B1);
  \draw[->, shorten <=0pt, shorten >=16pt] (x) -- (F1);

  \draw[-, gray, shorten <=1pt, shorten >=1pt] (B1) -- (D1);
  \draw[->, gray, shorten <=1pt, shorten >=1pt] (D1) -- (B2);
  \draw[-, gray, shorten <=1pt, shorten >=1pt] (B3) -- (D2);
  \draw[->, gray, shorten <=1pt, shorten >=1pt] (D2) -- (B4);
  \draw[-, gray, shorten <=1pt, shorten >=1pt] (B4) -- (D3);
  \draw[->, gray, shorten <=1pt, shorten >=1pt] (D3) -- (B5);
  
  \draw[->, gray, shorten <=5pt, shorten >=2pt] (B6) -- (B7);
  \draw[->, shorten <=5pt, shorten >=2pt] (F6) -- (F7);
  
  \draw[-, shorten <=1pt, shorten >=1pt] (F1) -- (DD1);
  \draw[->, shorten <=1pt, shorten >=1pt] (DD1) -- (F2);
  \draw[-, shorten <=1pt, shorten >=1pt] (F3) -- (DD2);
  \draw[->, shorten <=1pt, shorten >=1pt] (DD2) -- (F4);
  \draw[-, shorten <=1pt, shorten >=1pt] (F4) -- (DD3);
  \draw[->, shorten <=1pt, shorten >=1pt] (DD3) -- (F5);

  % switch arrows
  \draw[->, thick, shorten <=4pt, shorten >=18pt] (B2) -- (F3);
  \draw[->, thick, shorten <=4pt, shorten >=20pt] (B5) -- (F6);

\end{tikzpicture}
\caption{\textit{Example mass matrix adaptation scheme with background (below) and foreground (above) covariance estimators, with a switch/flush frequency of 10 draws. $\Sigma_i$ represents the covariance estimate used to sample draw $n+1$. Labels beneath them indicate the indices of the draw-gradient pairs on which that estimate is based. The transformation used for the Hamiltonian at each iteration is informed by the estimate stored in the foreground's state.}}
\label{fig:adapt_fig}
\end{figure}

\section[Details of the posteriordb evaluation]{\texttt{posteriordb} evaluation details}
\label{appendix:posteriordb}

Exceptionally slow-sampling models in \texttt{posteriordb} that we excluded from our analysis are\\\\
\texttt{ecdc0401-covid19imperial-v2},\\
\texttt{ecdc0401-covid19imperial-v3},\\
\texttt{ecdc0501-covid19imperial-v2},\\
\texttt{ecdc0501-covid19imperial-v3},\\
\texttt{election88-election88-full},\\
\texttt{hmm-gaussian-simulated-hmm-gaussian},\\
\texttt{iohmm-reg-simulated-iohmm-reg},\\
\texttt{mnist-nn-rbm1bJ100},\\
\texttt{mnist-100-nn-rbm1bJ10},\\
\texttt{prideprejudice-chapter-ldaK5},\\
\texttt{prideprejudice-paragraph-ldaK5}\\\\
The models which did not converge for any sampler/configuration are:\\\\
\texttt{GLMM-Poisson-data-GLMM-Poisson-model,\\
eight-schools-eight-schools-centered,\\ 
gp-pois-regr-gp-pois-regr,\\
low-dim-gauss-mix-collapse-low-dim-gauss-mix-collapse,\\
mcycle-gp-accel-gp,\\ 
mcycle-splines-accel-splines,\\
normal-2-normal-mixture,\\ 
normal-5-normal-mixture-k,\\
one-comp-mm-elim-abs-one-comp-mm-elim-abs,\\
ovarian-logistic-regression-rhs,\\ 
pilots-pilots,\\
prostate-logistic-regression-rhs,\\
radon-mn-radon-hierarchical-intercept-centered,\\
radon-mn-radon-variable-intercept-slope-centered,\\
radon-mn-radon-variable-slope-centered,\\ 
soil-carbon-soil-incubation,\\
surgical-data-surgical-model,\\
synthetic-grid-RBF-kernels-kronecker-gp,\\
three-docs1200-ldaK2,\\
three-men3-ldaK2,\\
uk-drivers-state-space-stochastic-level-stochastic-seasonal}\\\\
The remaining models were\\\\
\texttt{GLMM-data-GLMM1-model, GLM-Binomial-data-GLM-Binomial-model, GLM-Poisson-Data-GLM-Poisson-model, M0-data-M0-model, Mb-data-Mb-model, Mh-data-Mh-model, Mt-data-Mt-model, Mtbh-data-Mtbh-model, Mth-data-Mth-model, Rate-1-data-Rate-1-model, Rate-2-data-Rate-2-model, Rate-3-data-Rate-3-model, Rate-4-data-Rate-4-model, Rate-5-data-Rate-5-model, Survey-data-Survey-model, arK-arK, arma-arma11, bball-drive-event-0-hmm-drive-0, bball-drive-event-1-hmm-drive-1, bones-data-bones-model, butterfly-multi-occupancy, diamonds-diamonds, dogs-dogs, dogs-dogs-hierarchical, dogs-dogs-nonhierarchical, dugongs-data-dugongs-model, earnings-earn-height, earnings-log10earn-height, earnings-logearn-height, earnings-logearn-height-male, earnings-logearn-interaction, earnings-logearn-interaction-z, earnings-logearn-logheight-male, eight-schools-eight-schools-noncentered, fims-Aus-Jpn-irt-2pl-latent-reg-irt, garch-garch11, gp-pois-regr-gp-regr, hmm-example-hmm-example, hudson-lynx-hare-lotka-volterra, irt-2pl-irt-2pl, kidiq-kidscore-interaction, kidiq-kidscore-momhs, kidiq-kidscore-momhsiq, kidiq-kidscore-momiq, kidiq-with-mom-work-kidscore-interaction-c, kidiq-with-mom-work-kidscore-interaction-c2, kidiq-with-mom-work-kidscore-interaction-z, kidiq-with-mom-work-kidscore-mom-work, kilpisjarvi-mod-kilpisjarvi, loss-curves-losscurve-sislob, low-dim-gauss-mix-low-dim-gauss-mix, lsat-data-lsat-model, mesquite-logmesquite, mesquite-logmesquite-logva, mesquite-logmesquite-logvas, mesquite-logmesquite-logvash, mesquite-logmesquite-logvolume, mesquite-mesquite, nes1972-nes, nes1976-nes, nes1980-nes, nes1984-nes, nes1988-nes, nes1992-nes, nes1996-nes, nes2000-nes, nes-logit-data-nes-logit-model, radon-all-radon-county-intercept, radon-all-radon-hierarchical-intercept-centered, radon-all-radon-hierarchical-intercept-noncentered, radon-all-radon-partially-pooled-centered, radon-all-radon-partially-pooled-noncentered, radon-all-radon-pooled, radon-all-radon-variable-intercept-centered, radon-all-radon-variable-intercept-noncentered, radon-all-radon-variable-intercept-slope-centered, radon-all-radon-variable-intercept-slope-noncentered, radon-all-radon-variable-slope-centered, radon-all-radon-variable-slope-noncentered, radon-mn-radon-county-intercept, radon-mn-radon-hierarchical-intercept-noncentered, radon-mn-radon-partially-pooled-centered, radon-mn-radon-partially-pooled-noncentered, radon-mn-radon-pooled, radon-mn-radon-variable-intercept-centered, radon-mn-radon-variable-intercept-noncentered, radon-mn-radon-variable-intercept-slope-noncentered, radon-mn-radon-variable-slope-noncentered, radon-mod-radon-county, rats-data-rats-model, rstan-downloads-prophet, sat-hier-2pl, sblrc-blr, sblri-blr, science-irt-grsm-latent-reg-irt, seeds-data-seeds-centered-model, seeds-data-seeds-model, seeds-data-seeds-stanified-model, sesame-data-sesame-one-pred-a, sir-sir, state-wide-presidential-votes-hierarchical-gp, three-men1-ldaK2, three-men2-ldaK2, timssAusTwn-irt-gpcm-latent-reg-irt, traffic-accident-nyc-bym2-offset-only, wells-data-wells-daae-c-model, wells-data-wells-dae-c-model, wells-data-wells-dae-inter-model, wells-data-wells-dae-model, wells-data-wells-dist, wells-data-wells-dist100-model, wells-data-wells-dist100ars-model, wells-data-wells-interaction-c-model, wells-data-wells-interaction-model}\\\\
Details and code for these models can be found \href{https://github.com/stan-dev/posteriordb/blob/master/posterior_database/models/info}{here}. Some other interesting results from the \texttt{posteriordb} evaluation are found in \Cref{fig:supplemental_ecdfs}. 

\begin{figure}[H]
    \centering
    \includegraphics[width=0.95\textwidth]{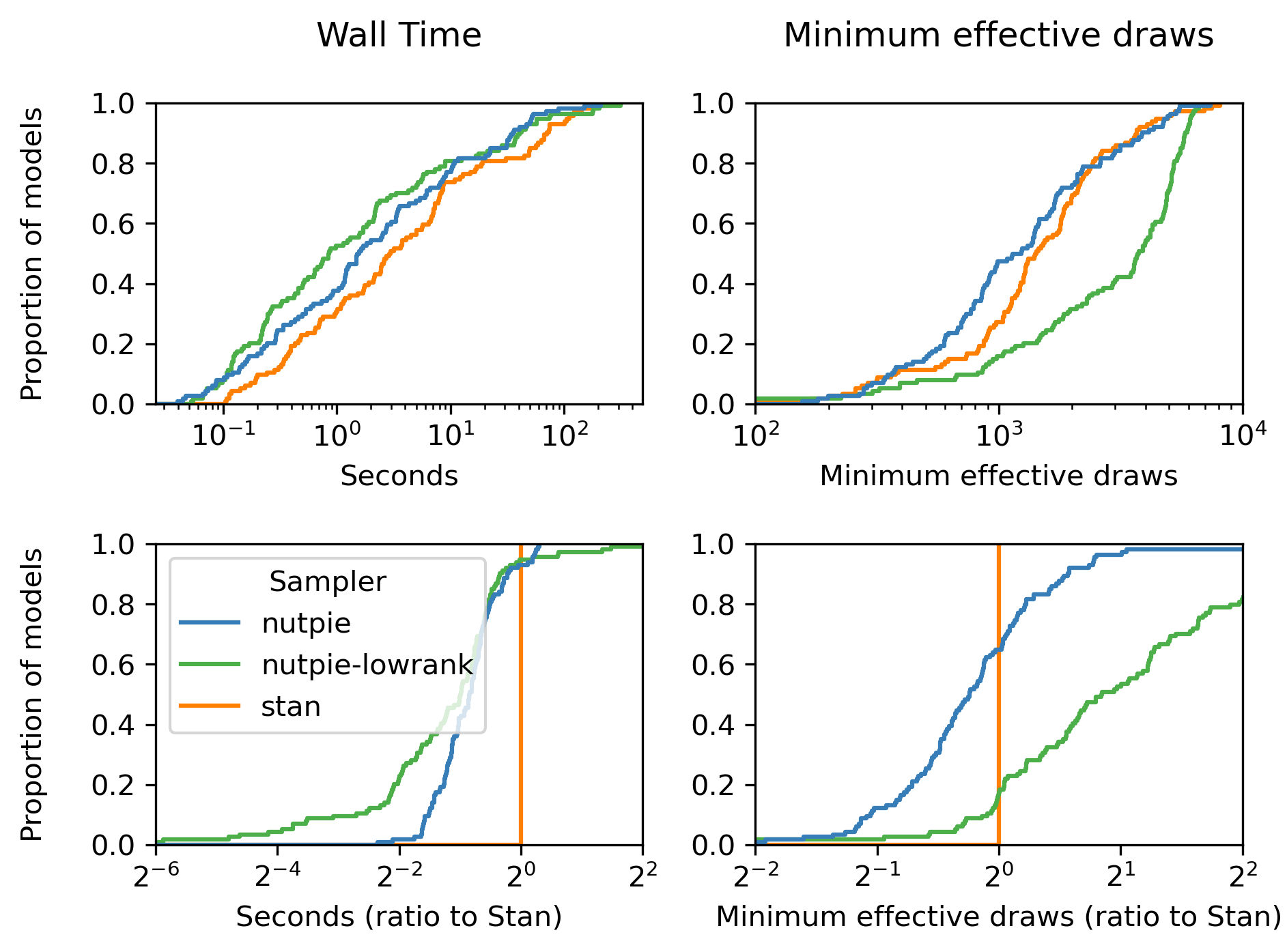}
    \caption{\textit{Cumulative distribution plots for target densities in \texttt{posteriordb}.  Each line represents a different sampler, \texttt{nutpie} (blue), \texttt{Stan} (orange), and \texttt{nutpie} low-rank (green) run for 1000 warmup iterations and 1000 posterior draws. The top row of plots show the raw diagnostic, while the bottom show the ratio of the diagnostic to Stan's. In the bottom left plot, the points to the right of 1 represent a small fraction of models for which Stan's default outperforms the other options in wall time. Interestingly, for a sizable chunk models, \texttt{nutpie} lags behind Stan in minimum effective draws. However, this is more than made up for in the reduction in the number of required gradient evaluations.}}
    \label{fig:supplemental_ecdfs}
\end{figure}

\begin{figure}[H]
    \centering
    \includegraphics[width=0.95\textwidth]{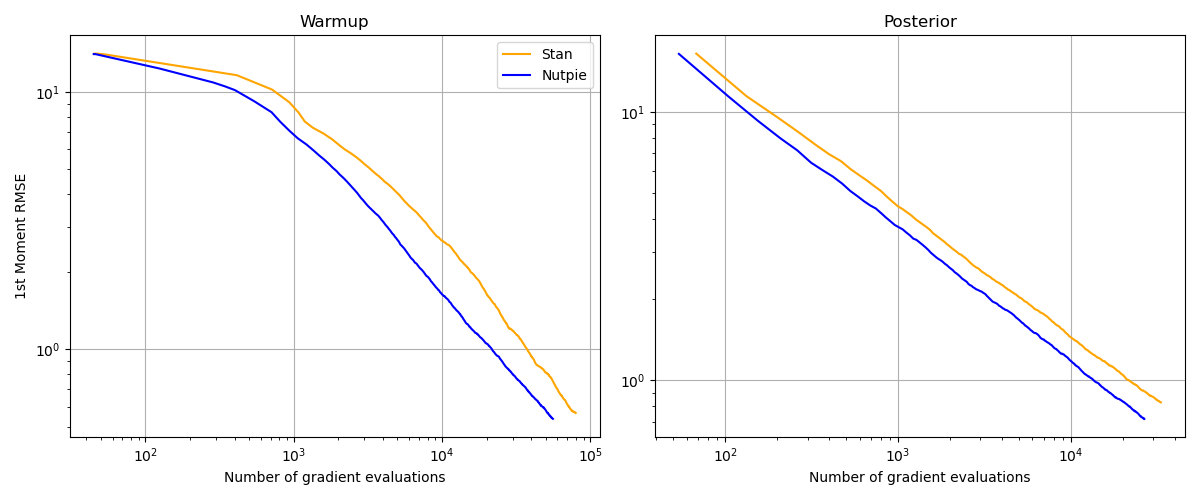}
    \caption{\textit{Average cumulative root mean squared error (RMSE) by gradient evaluations across 100 chains run for 1000 warmup draws and 500 posterior draws from a 100-D multivariate normal with a randomly generated $\Sigma = D^{1/2}U\diag{\lambda^2}U^TD^{1/2}$, $U\sim \text{Uniform}(O(200))$, $\lambda_i\sim \textrm{exp}(2)$, $D_i\sim \textrm{\emph{lognormal}}({0, 1})$. This procedure for generating $\Sigma$ produced an eigenspectrum with a maximum eigenvalue around 20 and most fewer than 1, leading to a high condition number $\kappa'$ of 68. Diagonal \texttt{nutpie} (blue) converged faster on average than \texttt{Stan} (red) in both warmup and stationary regimes, achieving a lower 1st moment RMSE in warmup with about 25,000 fewer gradient evaluations.}}
    \label{fig:rmse_comparisons}
\end{figure}

\begin{figure}[H]
    \centering
    \includegraphics[width=0.95\textwidth]{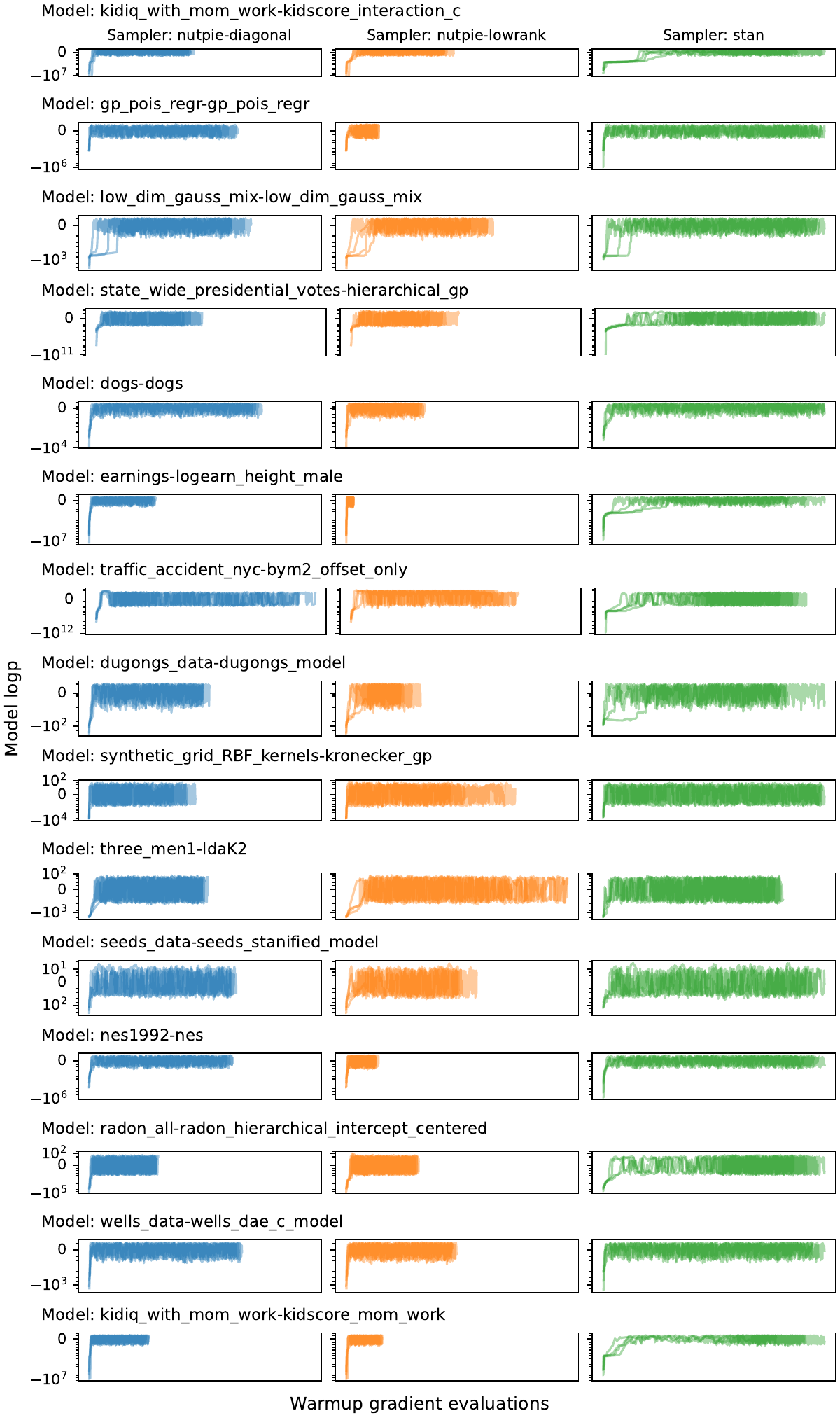}
    \caption{\textit{Trace-plots like \cref{fig:warmup-trace} of 15 randomly selected models from \texttt{posteriordb}.}}
    \label{fig:warmup-trace-random}
\end{figure}

\end{document}